\documentclass[10pt,reqno,a4paper,final]{amsart}
\usepackage[british]{babel}
\usepackage{amssymb,amstext,amsthm,eucal,bbm,amscd,bbold}
\usepackage[margin=1in]{geometry}
\usepackage{array}
\usepackage[svgnames,table]{xcolor}
\usepackage[T1]{fontenc}
\usepackage[utf8]{inputenc}
\usepackage[small]{eulervm}
\usepackage{mathrsfs}
\usepackage{verbatim}
\usepackage{pdfsync}
\usepackage[final]{microtype}
\usepackage[nosort]{cite}
\usepackage{adjustbox}
\usepackage{booktabs} 
\usepackage{caption}
\usepackage{appendix}
\usepackage{chngcntr}
\usepackage{apptools}
\usepackage{romanbar}
\usepackage{enumitem}
\usepackage{slashed}
\usepackage{tikz,pgfplots}
\usepgfplotslibrary{fillbetween}
\usetikzlibrary{calc,arrows,arrows.meta,cd,shapes.geometric,positioning,through,decorations.markings,decorations.text}
\tikzset{>=latex}
\usepackage[unicode,pagebackref=true]{hyperref}
\hypersetup{%
  pdftitle   = {BMS-like algebras: canonical realisations and BRST quantisation},
  pdfkeywords = {BMS, Lie algebras, operator product algebras, canonical realisation, BRST cohomology},
  pdfauthor  = {Carles Batlle, José Figueroa-O'Farrill, Joaquim Gomis,  Girish Vishwa},
  pdfcreator = {\LaTeX\ with package \flqq hyperref\frqq},
  linkcolor=NavyBlue,
  citecolor=ForestGreen,
  urlcolor=OrangeRed,
  anchorcolor=OrangeRed,
  colorlinks=true}
\renewcommand*{\backref}[1]{}
\renewcommand*{\backrefalt}[4]{%
  \ifcase #1%
  \or [Page~#2.]%
  \else [Pages~#2.]%
  \fi%
}
\usepackage{cleveref}
\theoremstyle{plain}
\newtheorem{lemma}{Lemma}
\newtheorem{proposition}[lemma]{Proposition}

\theoremstyle{definition}
\newtheorem{definition}[lemma]{Definition}
\newtheorem*{remark}{Remark}
\newcommand{\g}{\mathfrak{g}}

\newcommand{\h}{\mathfrak{h}}

\newcommand{\db}{\partial b}
\newcommand{\dc}{\partial c}
\newcommand{\ddb}{\partial^2 b}
\newcommand{\ddc}{\partial^2 c}

\newcommand{\so}{\mathfrak{so}}

\newcommand{\w}{\mathfrak{w}}

\newcommand{\Der}{\operatorname{Der}}
\newcommand{\reg}{\operatorname{reg}}

\renewcommand{\div}{\operatorname{div}}

\newcommand{\JJ}{\mathbb{J}}
\newcommand{\GG}{\mathbb{G}}
\newcommand{\TT}{\mathbb{T}}

\newcommand{\RR}{\mathbb{R}}

\newcommand{\NN}{\mathbb{N}}

\newcommand{\ZZ}{\mathbb{Z}}

\newcommand{\CC}{\mathbb{C}}
\newcommand{\1}{\mathbb{1}}

\newcommand{\Hom}{\operatorname{Hom}}

\newcommand{\dotr}{\mbox{$\boldsymbol{\cdot}$}}

\definecolor{dkgr}{rgb}{0,0.6,0}
\definecolor{gris}{rgb}{0.5,0.5,0.5}

\allowdisplaybreaks[0]
\numberwithin{equation}{section}
\AtAppendix{\counterwithin{lemma}{section}}

\def\vx{\ensuremath{\vec x}}

\def\vq{\ensuremath{\vec q}}
\def\vk{\ensuremath{\vec k}}

\newcommand{\dk}{d\widetilde{k}\,}

\def\d{\ensuremath{\text{d}}}
\newcommand{\Z}{\mathbb Z}
\begin{document}

\title{BMS-like algebras: canonical realisations and BRST quantisation}

\author[Batlle]{Carles Batlle}
\author[Figueroa-O'Farrill]{José M Figueroa-O'Farrill}
\author[Gomis]{Joaquim Gomis}
\author[Vishwa]{Girish S Vishwa}

\address[JMF,GSV]{Maxwell Institute and School of Mathematics, The University
  of Edinburgh, James Clerk Maxwell Building, Peter Guthrie Tait Road,
  Edinburgh EH9 3FD, Scotland, United Kingdom}
\address[CB]{Institut d'Organització i Control i Departament de
  Matemàtiques, Universitat Politècnica de Catalunya, EPSEVG,
  Av. V. Balaguer 1,  Vilanova i la Geltrú, 08800 Spain}
\address[JG]{Departament de Física Quàntica i Astrofísica and Institut
  de Ciències del Cosmos (ICCUB), Universitat de Barcelona, Martí i
  Franquès 1, Barcelona, 08028 Spain}

\email[JMF]{\href{mailto:j.m.figueroa@ed.ac.uk}{j.m.figueroa@ed.ac.uk}, ORCID: \href{https://orcid.org/0000-0002-9308-9360}{0000-0002-9308-9360}}
\email[GSV]{\href{mailto:}{G.S.Vishwa@sms.ed.ac.uk}, ORCID: \href{https://orcid.org/0000-0001-5867-7207}{0000-0001-5867-7207}}
\email[CB]{\href{mailto:}{carles.batlle@upc.edu}, ORCID: \href{https://orcid.org/0000-0002-6088-6187}{0000-0002-6088-6187}}
\email[JG]{\href{mailto:}{joaquim.gomis@ub.edu}, ORCID: \href{https://orcid.org/0000-0002-8706-2989}{0000-0002-8706-2989}}

\begin{abstract}
  We generalise BMS algebras in three dimensions by the introduction
  of an arbitrary real parameter $\lambda$, recovering the standard
  algebras (BMS, extended BMS and Weyl-BMS) for $\lambda=-1$.  We
  exhibit a realisation of the (centreless) Weyl $\lambda$-BMS algebra
  in terms of the symplectic structure on the space of solutions of
  the massless Klein-Gordon equation in $2+1$, using the eigenstates
  of the spacetime momentum operator. The quadratic Casimir of the
  Lorentz algebra plays an essential rôle in the construction. The
  Weyl $\lambda$-BMS algebra admits a three-parameter family of
  central extensions, resulting in the (centrally extended) Weyl-BMS
  algebra, which we reformulate in terms of operator product
  expansions.  We construct the BRST complex of a putative Weyl-BMS
  string and show that the BRST cohomology is isomorphic to the chiral
  ring of a topologically twisted $N=2$ superconformal field theory.
  We also comment on the obstructions to obtaining a conformal BMS Lie
  algebra -- that is, one that includes in addition the
  special-conformal generators -- and the need to consider a
  W-algebra.  We then construct the quantum version of this W-algebra
  in terms of operator product expansions.  We show that this
  W-algebra does not admit a BRST complex.
\end{abstract}

\maketitle
\tableofcontents

\section{Introduction and summary of the results}
\label{sec:introduction}

The BMS algebra was introduced in \cite{Bondi:1962px,Sachs:1962zza} as
the asymptotic symmetry algebra of a four-dimensional flat spacetime
at null infinity.  The BMS algebra extends the Poincaré algebra by an
infinite number of ``super-translations''.   Longhi and Materassi
\cite{Longhi:1997zt} found a canonical realisation of this algebra in
terms of the natural symplectic structure of the Fourier modes of a
free Klein--Gordon (KG) field in Minkowski spacetime.  A comprehensive
presentation of recent applications of BMS symmetries can be found in
\cite{Strominger:2017zoo}.

The BMS algebra was later extended to include an infinite number of
``super-rotations''
\cite{Barnich:2009se,Barnich:2010ojg,Barnich:2011mi} (see also
\cite{Fiorucci:2024ndw} for more recent developments).  There also
exists a canonical realisation of this extended BMS algebra in terms
of a free massless KG field in three-dimensional Minkowski spacetime
\cite{Batlle:2017llu}, expressed in terms of plane waves.  The
massless KG equation is conformally invariant, so it is natural to ask
whether there is a way to further extend the algebra by adding
(super-)dilatations and (super-)special-conformal transformations.
The answer for dilatations is positive and results in the Weyl--BMS algebra
\cite{Donnay:2020fof,Batlle:2020hia,Fuentealba:2020zkf,Freidel:2021fxf},
which further extends the extended BMS algebra by
superdilatations. The answer for special-conformal transformations
seems to be negative; although there is a W-algebra which may be
argued to extend the Weyl--BMS algebra \cite{Fuentealba:2020zkf}.

One of the main aims of this paper is to study a generalisation of
these algebras, in the case of $2+1$ dimensions, by the introduction
of an arbitrary real parameter $\lambda$, recovering the standard
algebras (BMS, extended BMS and Weyl-BMS) for $\lambda=-1$.  We call
them (Weyl) $\lambda$-BMS algebras, but they should not be confused
with the cosmological $\Lambda$-BMS algebras discussed, for instance,
in \cite{Compere:2019bua}.  Just like the BMS algebra can be described
intrinsically as the Lie algebra of infinitesimal automorphisms of the
conformal carrollian structure at null infinity in Minkowski spacetime
\cite{Duval:2014uva} (see also \cite{Figueroa-OFarrill:2019sex} for
other BMS-like algebras with a similar interpretation as infinitesimal
conformal automorphisms of other carrollian structures), the
$\lambda$-BMS algebras have been shown to appear as the near-horizon
algebra of non-extremal black hole horizons \cite{Grumiller:2019fmp}
(their spin $s$ and our $\lambda$ are related by a sign: $s =
-\lambda$.)

The paper is divided into three main sections.  In
Section~\ref{sec:cano-real} we will exhibit a realisation of the
(centreless) Weyl $\lambda$-BMS algebra in terms of the symplectic
structure on the space of solutions of the three-dimensional massless
Klein--Gordon equation.  Deriving inspiration from the case of the BMS
algebra (corresponding to $\lambda=-1$), we consider in
Section~\ref{sec:super-trans} the semidirect product of the Lorentz
Lie algebra with the super-translations.  The super-translations form
an infinite-dimensional abelian Lie algebra spanned by the
eigenfunctions of eigenvalue $\lambda(\lambda -1)$ of the quadratic
Casimir $C_2$ of the Lorentz Lie algebra, thought of as a second-order
differential operator on the smooth functions on the lightcone (i.e.,
the massless mass shell).

The Lorentz Lie algebra is realised as vector fields on the lightcone
and naturally act on the super-translations via the Lie derivative.  In
Section~\ref{sec:super-rots} we observe that not only the Lorentz Lie
algebra, but indeed any vector field on the lightcone which commutes
with $C_2$ also acts on the super-translations.  The additional
such vector fields are the super-rotations and they form a Lie algebra
isomorphic to the Witt algebra (i.e., the centreless Virasoro
algebra).  Together with the super-translations, one obtains the
(extended) $\lambda$-BMS algebra.  In
Section~\ref{sec:quadr-casim-lambda-BMS} we determine the quadratic
Casimirs of the $\lambda$-BMS algebra, which might be a result of
independent interest.

The massless Klein--Gordon field is not only Poincaré invariant, but
actually conformally invariant. In particular it is invariant under
dilatations, whose generator acts on functions as a differential
operator of degree at most one, of which the super-translations form
an eigenspace. We then ask whether there are other differential
operators of degree at most one which commute with $C_2$ and which act
on super-translations with the same eigenvalue as the dilatation. The
answer is positive and we obtain in this way a family of such
operators $D_n(k_1,k_2,k_3)$ depending on an integer parameter ($n$)
and three real parameters $(k_i)$.  The real parameters are fixed by
demanding that for $n=0$ we should recover the dilatation in the
conformal algebra.  This results in generators $D_n$, which we call
superdilatations. The resulting Lie algebra is the (centreless) Weyl
$\lambda$-BMS Lie algebra.

A natural question is whether one can do the same with the
special-conformal generators and extend them to a ``conformal''
$\lambda$-BMS algebra.  We argue in Section~\ref{sec:superconf} that
no such extension exists as a Lie algebra.  It is known, however, that
there is an extension as a W-algebra \cite{Fuentealba:2020zkf}, which
is discussed in Section~\ref{sec:conformal-bms-w}.

In Section~\ref{sec:weyl-lambda-bms} we show that the Weyl
$\lambda$-BMS algebra admits a three-parameter family of central
extensions, resulting in the (centrally extended) Weyl--BMS algebra.
We then reformulate the centrally extended algebra in terms of
operator product expansions.  In that language we then proceed to
construct the BRST complex of putative Weyl--BMS strings, showing that
that it exists provided the central charges are chosen judiciously.
We then show that the BRST cohomology is isomorphic to the chiral ring
of a topologically twisted $N=2$ superconformal theory obtained by
coupling the Weyl--BMS string to topological gravity in the form of a
Koszul topological conformal field theory.  This provides further
evidence for a conjecture in
\cite{Figueroa-OFarrill:1995agp,Figueroa-OFarrill:1996dic} that the 
BRST cohomology of every topological conformal field theory is
isomorphic to the chiral ring of an $N=2$ superconformal field theory.

In Section~\ref{sec:conformal-bms-w} we return to the case of $\lambda
= -1$ and we first of all construct the fully quantum conformal BMS
W-algebra of \cite{Fuentealba:2020zkf} in terms of operator product
expansions and then report on calculations showing that, perhaps
contrary to expectations, there is no BRST complex for this W-algebra.

The paper ends with a short Section~\ref{sec:conclusions-outlook} with
conclusions and a look at future extensions of this work and two
appendices with technical results used in Section~\ref{sec:cano-real}.

\section{Canonical realisations}
\label{sec:cano-real}

In this section we will consider a canonical realisation of a
generalisation of the Weyl-BMS algebra that depends on one
parameter. We will see the crucial role of the Casimir of the Lorentz
group. We will start by including only the super-translations, and
their $\lambda$-depending generalisations, and then will proceed to
include the super-rotations and the super-dilatations.

The Lagrangian density for a real massless scalar field in flat
Minkowski space time\footnote{The signature of the Minkowski metric is
  $(-++)$} is
\begin{equation}
  \mathcal{L} = -\dfrac{1}{2}\partial_\mu \phi \partial^\mu \phi.
\end{equation}
The solution to the equation of motion, Klein--Gordon equation, in terms of the Fourier modes $a(\vk)$,
\begin{equation}
  \phi(t,\vx) = \int\dk \left(a(\vk)e^{ikx}+ \bar{a}(\vk)e^{-ikx}\right),
\end{equation}
with $x=(t,\vx)$, $kx=-\omega t+ \vk\cdot\vx$,
$\omega=k^0=\sqrt{\vk^2}$.   The solution is expressed in terms of
the plane waves that are eigenstates of the momentum operator and
\begin{equation}
  \dk = \frac{\d^2 k}{\Omega(\vk)},\quad \Omega(\vk) = (2\pi)^2 2 \omega,
\end{equation}
and where the Fourier modes satisfy the non-zero Poisson brackets
\begin{equation}
  \{a(\vk), \bar{a}(\vq)  \} = -i \Omega(\vk)\delta^2(\vk-\vq).
\end{equation}
Notice that we parameterise the mass-shell manifold of the massless
scalar-field, $k^2=0$, using $\vec{k}\in\RR^2$.  Alternatively, we
could expand the solution in terms of other eigenfunctions, for
example using the eigenfunctions of a boost generator, which would
lead in a natural way to celestial holography
(see, e.g., \cite{Pasterski:2021rjz,Donnay:2023mrd}) expressing
the four dimensional theory in terms of two-dimensional complex
conformal field theory. In our case, in three dimensions, a chiral
two-dimensional conformal field theory seems to appear naturally.

The conserved charges associated to the translation and Lorentz
generators are
\begin{align}
  P^\mu &= \int\dk\bar{a}(\vk) k^\mu a(\vk),\quad \mu=0,1,2, \label{sPmu}\\
  M^{ij} &= -i \int\dk \bar{a}(\vk)\left(k^i\frac{\partial}{\partial k^j}-k^j \frac{\partial}{\partial k^i}  \right) a(\vk),\quad j=1,2\label{sMij}\\
  M^{0j} &= tP^j - i \int\dk \bar{a}(\vk) k^0 \frac{\partial}{\partial k^j} a(\vk), \quad j=1,2, \label{sM0j}
\end{align}
while the charge associated to dilatations is
\begin{equation}
  D = -tP^0 + i \int\dk\, \bar{a}(\vk) \left( k^j \frac{\partial}{\partial k^j} + \dfrac{1}{2} \right) a(\vk).
\end{equation}

From the above equations one can read, at $t=0$,
the differential operators for boosts, rotation and dilatation:   
\begin{align}
  \hat{B}_j &= i \omega \frac{\partial}{\partial k^j},\ j=1,2,\\
  \hat{J} &= -i \left(k^1\frac{\partial}{\partial k^2}-k^2 \frac{\partial}{\partial k^1}  \right),\\
  \hat{D} &= i\left( k^j\frac{\partial}{\partial k^j}+ \frac{1}{2}\right),
\end{align}
with $B_j=M_{0j}$. As shown in Appendix \ref{AppCharges}, the algebra
of the Poisson brackets of the charges induced by the symplectic
structure of the Fourier coefficients is the same, up to a $-i$
factor, of that of the associated differential operators, and we will
work with the later.

Using polar coordinates in the $(k^1,k^2)$ plane, $k^1=r\cos\phi$, $k^2=r\sin\phi$, $\omega=r$, one has
\begin{align}
	\hat{B}_1 &= i r \cos\phi\ \partial_r - i \sin\phi\ \partial_\phi,\\
	\hat{B}_2 &= i r \sin\phi\ \partial_r + i \cos\phi\ \partial_\phi,\\
	\hat{J} &= -i \partial_\phi,\\
	\hat{D} &=i\left(r\partial_r+\frac{1}{2}\right).\label{opD}
\end{align} 
In terms of these operators, the quadratic Casimir of the Lorentz
group in $2+1$, $\frac{1}{2}M^{\mu\nu}M_{\mu\nu}$, is
\begin{equation}
  \hat{C}_2 = -\hat{B}_1^2 - \hat{B}_2^2 + \hat{J}^2 = r^2\partial_r^2 + 2r\partial_r.
  \label{Cas1}
\end{equation}

In this realisation, there is a fundamental relation between the
Casimir and the square of the dilatation generator,
\begin{equation}
  \hat{C}_2 = - \hat{D}^2 - \left(\frac{1}{2}\right)^2.
  \label{CD2}
\end{equation}

\subsection{Super-translations and $\lambda$-BMS algebras}
\label{sec:super-trans} 

Following  \cite{Delmastro:2017erq}, which generalises to arbitrary
dimensions the results in \cite{Longhi:1997zt}, we can construct
BMS-like algebras by considering the solutions of 
\begin{equation}
	-\hat{C}_2 \xi = \alpha\xi,
	\label{eigen2}
\end{equation}
where the minus sign is added in order to more easily connect with the
standard notation of the representations of the Lorentz group. As
shown in Appendix \ref{AppPoincare}, the solutions to (\ref{eigen2})
provide representations of the Lorentz algebra.

The Casimir eigenvalue equation (\ref{eigen2}) in terms of polar coordinates is
\begin{equation}
r^2 \partial_r^2\xi +2 r \partial_r\xi = -\alpha \xi,
\label{Casim1}
\end{equation}
Since this does not depend on the angular coordinate, the solutions
will be of the form  $\xi_n(r,\phi)=f(r)e^{in\phi}$, with
$n\in\Z$. Looking for radial solutions of the form
\begin{equation}
  f(r) = r^\beta
\end{equation}
one finds
\begin{equation}
	\beta= \frac{-1\pm\sqrt{1-4\alpha}}{2}=
	\frac{1}{2}(-1\pm (1-2\lambda)),
\end{equation}
where we have defined
\begin{equation}
  1-4\alpha = (1-2\lambda)^2.
\end{equation}
with inverse relation
\begin{equation}
  \alpha = -\lambda (\lambda-1).
\end{equation}
 
One gets thus the two families of solutions, with $\beta=-\lambda$ and
$\beta=\lambda-1$. In order to get all possible values of
$\beta\in\RR$, it is enough to take $\lambda\in\RR$ and consider only
the solutions $\beta=-\lambda$.\footnote{One can also consider complex
  values of $\lambda$ provided that the real part is $1/2$, since in
  this case $\alpha$ remains real. In the representation theory of
  $SO(2,1)$ this corresponds to the unitary principal series
  representation (see, e.g., \cite{Sun:2021thf}). However, we will not
  pursue this possibility in this work.}

The complete solution, with the $S^1$ angular coordinate, will  be
\begin{equation}
	\omega_n(r,\phi) = r^{-\lambda} e^{in\phi}, \quad n\in\Z,
\end{equation}
and then
\begin{equation}
	\hat{C}_2 \omega_n = \lambda(\lambda-1)  \omega_n.
\end{equation}

These $\omega_n$ can be put in (\ref{sPmu}) instead of $k^\mu$ to
obtain the corresponding conserved charges. For $\lambda=-1$ one gets
the charges corresponding to the standard BMS transformations, which
can then be used to define the transformations of the fields in the
canonical formalism. As discussed in \cite{Longhi:1997zt} and
\cite{Batlle:2017llu,Batlle:2022hwf}, these transformations are in
general non-local. See also Appendix A of \cite{Banerjee:2020qjj} for
a more general discussion in this regard.

The action of the Lorentz generators on $ \omega_n(r,\phi)$ 
is
\begin{equation}
  \begin{split}
    \hat{B}_1 \omega_n &= -\frac{i}{2}(n+\lambda)\omega_{n+1} + \frac{i}{2} (n-\lambda)\omega_{n-1},\\
    \hat{B}_2 \omega_n &= -\frac{1}{2}(n+\lambda)\omega_{n+1} - \frac{1}{2} (n-\lambda)\omega_{n-1},\\
    \hat{J}\omega_n &= n \omega_n,
  \end{split}
\end{equation}
These equations provide an infinite-dimensional realisation of the
$2+1$ Lorentz algebra in the space of the
$\{\omega_n\}_{n\in\Z}$. Looking at the zeros of the coefficients
appearing in the above equations, several cases, depending on the
value of $\lambda$, can be considered:
\begin{enumerate}
\item If $\lambda\notin\ZZ$, then the coefficients in the above
  equations can never be zero and the representation is
  irreducible. This corresponds to the complementary series in the
  standard representation theory of $SO(2,1)$ (see, e.g.,
  \cite{Sun:2021thf}).
  
\item In $\lambda = -N$, $N\in\NN$, then one has a finite
  representation in the space $\{\omega_n\}_{|n|\leq N}$. In
  particular, for $\lambda=-1$ one obtains the vector representation
  in the space of the $\{\omega_{-1},\omega_0,\omega_{1}\}$.
  
\item If $\lambda=0$, one obtains the trivial representation spanned
  by $\{\omega_0\}$.
  
\item If $\lambda=N$, $N\in\NN$, there appear two infinite-dimensional
  representations, spanned by $\{\omega_n\}_{n\leq -N}$ and
  $\{\omega_n\}_{n\geq N}$, respectively, which correspond to the
  highest and lowest weight representations of the standard
  literature.
\end{enumerate}

Up to this point, we have only talked about representations of the
Lorentz group, which are provided by the set of $\omega_n$.  However,
since the Lorentz operators are of first order (without zero order
term) and the $\omega_n$ are of zeroth order, the above expressions
coincide with the commutators
\begin{align}
	[\hat{B}_1, \omega_n] &= -\frac{i}{2}(n+\lambda)\omega_{n+1} + \frac{i}{2} (n-\lambda)\omega_{n-1},\label{LBMS1}\\
	[\hat{B}_2, \omega_n] &= -\frac{1}{2}(n+\lambda)\omega_{n+1} - \frac{1}{2} (n-\lambda)\omega_{n-1},\label{LBMS2}\\
	[\hat{J},\omega_n] &= n \omega_n.\label{LBMS3}
\end{align}
These commutators, together with those between $\hat{B}_1$, $\hat{B}_2$ and $\hat{J}$, and adding the trivial ones between zero order operators
\begin{equation}
  [\omega_n,\omega_m] =0
\end{equation}
provide a realisation of an infinite dimensional algebra which, for
$\lambda=-1$, is the standard $BMS$ algebra in $2+1$, which contains
the finite Poincaré algebra obtained by considering the subset
$\{\omega_{-1},\omega_0, \omega_1\}$. For general $\lambda\in\RR$ we
obtain what we will call $\lambda$-BMS algebras.

 As is the case for the representations of Lorentz, for $\lambda=-N$,
 $N\in\NN$, one can obtain finite  subalgebras of dimension
 $3+(2N+1)$, with generators $\hat{B}_1$, $\hat{B}_2$, $\hat{J}$ and
 $\omega_{-N},\ldots,\omega_N$, which we call $\lambda$-Poincaré.

Notice that $\lambda=-1$, which is the value that leads to the
Poincaré subalgebra in $d=3$ spacetime, corresponds to $\alpha=2$ in
equation (\ref{Casim1}), which is equal to $d-1$ for $d=3$. As shown
in  Appendix \ref{AppPoincare}, $d-1$ is the eigenvalue in
(\ref{Casim1}) that makes Poincaré appear as a subalgebra of
$\lambda$-Poincaré for arbitrary spacetime dimension $d$.

It is customary to write the above algebras in terms of\footnote{For
  the rest of this section we will suppress the $\,\hat{}\,$ on the
  differential operators.} $L_0=-J$ and the ladder operators $ {L}_1=
-i{B}_1+{B}_2$, $ {L}_{-1}=i{B}_1+{B}_2$, which are explicitly given
by
\begin{align}
  {L}_1 &= i e^{i\phi}(\partial_\phi-ir\partial_r),\label{Lp}\\
  {L}_{-1} &= i e^{-i\phi}(\partial_\phi+ir\partial_r),\label{Lm}\\
  {L}_0 &= i\partial_\phi.\label{L0}
\end{align}
 
One has then the $\lambda$-algebra in the form
\begin{align}
  & [ {L}_1, {L}_{-1}] = 2 {L}_0,\quad [ {L}_0, {L}_1] = - {L}_1, 
    \quad [ {L}_0,{L}_{-1}]={L}_{-1},\label{LorentzL}\\
  &[ {L}_1,\omega_n] = -(n+\lambda)\,\omega_{n+1},\label{L+P}\\
  &[ {L}_{-1},\omega_n] = - (n-\lambda)\,\omega_{n-1},\label{L-P}\\
  & [L_0,\omega_n] = -n\omega_n, \label{L-0}\\
  & [\omega_n,\omega_m] =0.
\end{align}

In terms of these operators, the Lorentz Casimir can be expressed as
\begin{equation}
  C_2 = - L_1 L_{-1} + L_0^2 + L_0.
  \label{LC_LLL}
\end{equation}

It should be noticed that, since the Lorentz Casimir $C_2$ is a second
order operator, the commutator of $C_2$ with $\omega_n$ differs from
the action of $C_2$ on $\omega_n$, yielding a first order operator
(with zeroth term) instead of a function:
\begin{equation}\label{C2w}
  \begin{split}
    [C_2,\omega_n] &= [r^2\partial_r^2 + 2r\partial_r, r^{-\lambda} e^{in\phi}] = \lambda(\lambda-1) r^{-\lambda} e^{in\phi} - 2 \lambda r^{-\lambda+1} e^{in\phi}\partial_r\\
    &= C_2\omega_n - 2 \lambda \omega_n E,
  \end{split}
\end{equation}
where $E=r\partial_r$ is the Euler operator, which in Cartesian
coordinates has the expression $E=k_1\partial_{k_1}+k_2\partial_{k_2}$.

\subsubsection{Quadratic Casimirs of the $\lambda$-BMS algebras}
\label{sec:quadr-casim-lambda-BMS}

Since the quadratic Casimir of the Lorentz algebra does not commute
with the $\omega_n$, a general quadratic Casimir of the $\lambda$-BMS
algebra can only involve the $\omega_n$ generators, and will be of the
form
\begin{equation}
  C_2^\lambda = \sum_{n\in\Z}\sum_{m\in\Z} A_{nm}\, \omega_n\omega_m, \quad A_{nm}=A_{mn}.
\end{equation}
Demanding $[L_0,C_2^\lambda]=0$ leads to that $A_{nm}$ can be different from zero only if $n+m=0$, so that
\begin{equation}
  C_2^\lambda = \sum_{n\in\Z}  A_{n}\, \omega_n\omega_{-n}, \quad A_{n}=A_{n (-n)}, 
\end{equation}
with $A_{n}=A_{-n}$ due to the symmetry of the $A_{nm}$.  Imposing
$[L_1,C_2^\lambda]=0$, re-arranging terms and using $A_{n}=A_{-n}$,
one obtains the set of equations
\begin{equation}
  (-\lambda-n)A_n  +  (-\lambda+n+1)A_{n+1} =0,\quad n=0,1,2,\ldots,
  \label{casimir_rec}
\end{equation}
which are the same relations that are obtained imposing
$[L_{-1},C_2^\lambda]=0$. For $\lambda\notin\ZZ$, the recurrence
relations (\ref{casimir_rec}) can be solved for all the $A_n$ in terms
of $A_0$, and the Casimir has infinite terms. We will discuss
(\ref{casimir_rec}) for $\lambda\in\ZZ$, separating the cases
$\lambda<0$ and $\lambda>0$ (the case $\lambda=0$ is obviously
trivial).

\begin{itemize}
\item If $-\lambda=N\in\NN$, equations (\ref{casimir_rec}) become
  \begin{equation}
    (N-n)A_n + (N+n+1)A_{n+1} =0, \quad n=0,1,2,\ldots
  \end{equation}
  which can be solved for $A_1, A_2,\ldots, A_N$ in terms of $A_0$, and one obtains
  \begin{equation}
    A_n = (-1)^n \frac{N (N-1)\cdots (N-n+1)}{(N+n)(N+n-1)\cdots(N+1)}A_0,\quad n=1,2,\ldots,N.
  \end{equation}
  However, the equation for $n=N$ is just
  \begin{equation*}
    0\cdot A_N + (2N+1) A_{N+1} =0,
  \end{equation*}
  which implies $A_{N+1}=0$, and subsequently also $A_n=0$ for all
  $n>N$. Thus, using that $A_n=A_{-n}$ and taking $A_0=1$, the Casimir
  boils down to
  \begin{equation}
    C_2^\lambda = \omega_0^2 + 2 \sum_{n=1}^N (-1)^n \frac{N (N-1)\cdots (N-n+1)}{(N+n)(N+n-1)\cdots(N+1)} \omega_n\omega_{-n},
    \label{CasimirL}
  \end{equation}
  which will also be the Casimir for the finite dimensional $\lambda$-Poincaré algebra.

  For instance, for $\lambda=-1$ one obtains the well-known Poincaré quadratic Casimir 
  \begin{equation}
    C_2^{\lambda=-1} = \omega_0^2 - \omega_1\omega_{-1},
  \end{equation} 
  while for $\lambda=-2$ one gets
  \begin{equation}
    C_2^{\lambda=-2} = \omega_0^2 - \frac{4}{3} \omega_1 \omega_{-1} + \frac{1}{3} \omega_{2}\omega_{-2}.
  \end{equation}

\item If $\lambda=N\in\NN$ the recurrence relation is
  \begin{equation}\label{eq:recurrence}
    (-N-n)A_n + (-N+n+1)A_{n+1} =0,\quad n=0,1,2,\ldots
  \end{equation}
  For $n=N-1$ one obtains
  \begin{equation}
    (-2N+1)A_{N-1} +0\cdot A_N  =0,
  \end{equation} 
  so that $A_{N-1}=0$, 
  which then forces $A_{N-2}=\ldots=A_1=A_0=0$, and the resulting Casimir has infinite terms
  \begin{equation}
    \tilde{C}_2^\lambda  = 2\sum_{n=N}^\infty A_n {\omega}_n {\omega}_{-n},
  \end{equation}
  with all the $A_n$ computed in terms of $A_{N}$ using the recurrence
  relation~\eqref{eq:recurrence}, which can be solved for the generating
  function $A(z) = \sum_{n=0}^\infty A_{n+N} z^n$ as follows:
  \begin{equation}
    A(z) = \frac{A_N}{(1-z)^{2N}}.
  \end{equation}
\end{itemize}

\subsection{Super-rotations and extended $\lambda$-BMS algebras}
\label{sec:super-rots}

In the same way that super-translations are obtained by computing all
the solutions of the Lorentz Casimir $C_2$ eigenvalue equation, one
may try to obtain generalisations of the Lorentz generators by
computing all the first order differential operators that commute with
$C_2$.

Using polar coordinates in the massless mass-shell manifold we look for an operator
\begin{equation}
	L(r,\phi) = a(r,\phi)\partial_r + b(r,\phi)\partial_\phi
\end{equation}
such that $[L,C_2]=0$, with $C_2$ the second order differential operator in (\ref{Cas1}). One has
\begin{align}
	[L,C_2] &= \left(2ar-2r^2 \partial_r a \right) \partial_r^2 \nonumber\\
	&+ \left(2a-2r \partial_r a - 2 r\partial_r b - r^2 \partial_r^2 a \right) \partial_r \nonumber\\
	&- 2 r^2 \partial_rb \partial_r\partial_\phi - r^2 \partial_ r^2 b \partial_\phi.
\end{align}
The cancellation of the term in $\partial_r\partial_\phi$ forces $\partial_r b=0$, so that $b=b(\phi)$. This also cancels the term in $\partial_\phi$, and demanding that the two remaining terms are zero leads to $\partial_r^2 a =0$ and to
\begin{equation}
	r \partial_r a = a,
\end{equation}
with general solution $a(r,\phi) = r c(\phi)$, which also satisfies $\partial_r^2 a =0$. Hence, the most general first-order differential operator commuting with the Lorentz Casimir is 
\begin{equation}
	L= r c(\phi) \partial_r + b(\phi) \partial_\phi.
\end{equation}
Besides commuting with the Casimir, the differential operators
associated to the Lorentz generators are also divergenceless. Since
the Lorentz-invariant volume form on the massless mass-shell manifold
is proportional to $\d r\wedge\d \phi$, the divergence of a field in
polar coordinates is just the sum of the partial derivatives (see also
appendix B in \cite{Batlle:2017llu}), and we have
\begin{equation}
  0= \div  L = c(\phi) + \partial_\phi b(\phi),
\end{equation}
from which $c(\phi) = -b'(\phi)$. The first order differential operators that share all the relevant properties with the Lorentz generators are thus
\begin{equation}
	L= - r  b'(\phi)\partial_r + b(\phi) \partial_\phi.
\end{equation}
Since $\phi\in S_1$, we can expand in Fourier series and obtain an infinite set of operators $L_n$ indexed by $n\in\ZZ$. Writing $b(\phi) = i e^{in\phi}$, the resulting operators, called super-rotations \cite{Barnich:2009se}, are
\begin{equation}
	L_n = i e^{in\phi} (\partial_\phi - inr \partial_r),\quad n\in\ZZ,
	\label{superL}
\end{equation}
which coincide with the standard Lorentz generators for $n=-1,0,1$.

Adding the super-rotations one gets the extended $\lambda$-BMS algebras
\begin{align}
	[L_n,L_m] &= (n-m)L_{n+m},\label{extLL}\\
	[L_n, \omega_m] &= -(m+\lambda n)\omega_{n+m},\label{extLP}\\
	[\omega_n,\omega_m] &= 0,\label{extPP}
\end{align}
with $n,m\in\ZZ$. This algebra is the semi-direct sum of the Witt
algebra with its tensor density modules (see section
\ref{sec:weyl-lambda-bms}). Following the notation of
\cite{Figueroa-OFarrill:2024wgs}, we refer to it as $\g_\lambda$. It
is a special case of the $W(a,b)$ algebras \cite{MR2773310}, where
$a\in\ZZ$ and $b=\lambda$. Setting $\lambda=-1$ recovers the
centreless BMS algebra, whose most general deformation is shown to be
$W(a,b)$\footnote{We refer the reader to \cite{Banerjee:2023yln,
    McStay:2023thk} for some appearances of $W(0,-2)$ in physics.}
(even when central extensions are included)
\cite{FarahmandParsa:2018ojt}.

From (\ref{extLP}) one also obtains
\begin{equation}
  L_n\omega_m = -(m+\lambda n)\omega_{n+m},\label{RWitt}
\end{equation}
which provides, for each $\lambda\in\RR$, a representation of the Witt algebra.

Notice that, using (\ref{extLL}) and the Lorentz Casimir in the form  (\ref{Cas1})
\begin{equation}
  \begin{split}
    [L_n, C_2] &= -(n-1)L_{n+1}L_{-1} - (n+1) L_1 L_{n-1} + n L_nL_0 + nL_0 L_n + n L_n \\
    &= - (n-1) L_{n+1}L_{-1}- (n+1) L_{n-1} L_1 + 2n L_n L_0 - 2L_n.
  \end{split}
\end{equation}
By construction of the $L_n$ this must be zero, and the following identity must hold:
\begin{equation}
  - (n-1) L_{n+1}L_{-1}- (n+1) L_{n-1} L_1 + 2n L_n L_0 - 2L_n = 0, \quad n\in\ZZ.
  \label{CasLIden}
\end{equation}
This can be checked directly by using the explicit form
(\ref{superL}).\footnote{As a vector in the universal enveloping
  algebra $\mathfrak{U}$ of the Witt algebra, $(n-1) L_{n+1}L_{-1} +
  (n+1) L_{n-1} L_1 - 2n L_n L_0 + 2L_n \neq 0$, but of course it is
  in the kernel of the algebra homomorphism from $\mathfrak{U}$ to
  differential operators on the punctured plane acting on smooth
  functions.  The kernel of this homomorphism has been calculated in
  \cite[Theorem~1.2]{MR2320470} and it is a principal two-sided ideal
  of $\mathfrak{U}$ generated by $Z_1 := \frac12 (L_2 L_{-1} - L_1 L_0
  - L_1)$.  A quick calculation shows that $[L_{-1}, Z_1] = C_2$ and
  therefore $[L_n,C_2] = [L_n, [L_{-1}, Z_1]] = L_n L_{-1} Z_1 - L_n
  Z_1 L_{-1} - L_{-1} Z_1 L_n + Z_1 L_{-1} L_n$ indeed belongs to the
  two-sided ideal $\mathfrak{U} Z_1 \mathfrak{U}$.}

\subsection{Super-dilatations and the centreless Weyl $\lambda$-BMS algebras}
\label{sec:super-dils}

Consider a general differential operator of order up to one
\begin{equation} 
  {\mathbb D}=a\partial_ r + b\partial_\phi + c.
\end{equation}
Demanding that the $\omega_n$ provide a representation of ${\mathbb D}$,
\begin{equation}
  {\mathbb D} \omega_n = \alpha \omega_m,
\end{equation}
with $\alpha$ possibly depending on $m,n$,
one immediately gets
\begin{equation}
  a= k_1 r e^{i(m-n)\phi},\quad b=k_2 e^{i(m-n)\phi},\quad c= k_3 e^{i(m-n)\phi},
\end{equation}
with $k_i$ constants, and then
\begin{equation}
  \alpha = -k_1 \lambda + i k_2 n + k_3.
\end{equation}
One has then the operators
\begin{equation}
  D_{m-n} = e^{i(m-n)\phi} (k_1 r \partial_r + k_2 \partial_\phi + k_3),
\end{equation}
or, setting $p=m-n$,
\begin{equation}
  D_p = e^{ip\phi} (k_1 r \partial_r + k_2 \partial_\phi + k_3),
\end{equation}
which yields
\begin{equation}
  D_p \omega_n = ( -k_1 \lambda + i k_2 n + k_3)\omega_{n+p}.
\end{equation}
Demanding $D_0=D$ fixes $k_1=1$, $k_2=0$ and $k_3=1/2$, and one gets
\begin{equation}
  D_n =  e^{in\phi} \left(r\partial_r +\frac{1}{2}\right), \quad n\in\ZZ,
\end{equation}
which, for $n=0$ gives $D_0= -i D$, with $D$ the standard dilatation
operator defined in (\ref{opD}). Up to an
overall constant factor, these agree with the super-dilatation operators
introduced in \cite{Batlle:2020hia}.

Adding the commutators of $D_n$ with the super-translations and
super-rotations one obtains the centreless Weyl $\lambda$-BMS algebra:
\begin{align}
	[L_n,L_m] &= (n-m)L_{n+m},\label{WextLL}\\
	[L_n, P_m] &= -(m+\lambda n)P_{n+m},\label{WextLP}\\
	[P_n,P_m] &= 0,\label{WextPP}\\
	[L_n,D_m] &= -m D_{m+n},\label{WextLD}\\
	[D_n, P_m] &= -\lambda P_{m+n},\label{WextDP}\\
	[D_n, D_m] &= 0,\label{WextDD}
\end{align}
where we have nominally replaced the super-translation functions
$\omega_n$ with the tensor densities\footnote{See Section
  \ref{sec:weyl-lambda-bms}.} $P_n$ since, as shown by (\ref{WextDP}),
the $\omega_n$ do not transform as true functions under
super-dilatations, unless $\lambda=0$.

Notice that (in the case $\lambda \neq 0$) by redefining $D_n
\longrightarrow -D_n/\lambda$ we can set the coefficient of the
right-hand side of (\ref{WextDP}) to unity without changing the other
commutation relations.  We note that the algebra is not centrally
extended in this realisation.  We will see in
Section~\ref{sec:weyl-lambda-bms} that this algebra admits a
three-parameter family of central extensions (see also
\cite{Freidel:2021fxf} for the case of $\lambda=-1$).

\subsection{BMS-like super-special-conformal generators}
\label{sec:superconf}

The conformal algebra adds generators $D$ (dilatations) and $K^\mu$
(special-conformal transformations) to the Poincaré algebra. As shown in
\cite{Batlle:2020hia}, the corresponding differential
operators\footnote{We work here in Cartesian coordinates since the
  expression of the special-conformal operators in polar coordinates
  is not particularly simple.} in $2+1$ spacetime acting on the
mass-shell manifold of a massless scalar field are, besides $D$,
\begin{align}
	K^0 &= \omega \frac{\partial^2}{\partial k_i\partial k_i} ,\\
	K^j &= \left( k_j \frac{\partial}{\partial k_i} - 2 k_i \frac{\partial}{\partial k_j}\right)  \frac{\partial}{\partial k_i} - \frac{\partial}{\partial k_i},\quad i=1,2,
\end{align}
and the commutators involving them are, besides those of the Poincaré algebra,
\begin{align}
	[D,K^\mu] &= i K^\mu,\label{confDK}\\
	[K^\mu,M^{\nu\sigma}] &= -i (\eta^{\mu\sigma}K^\nu - \eta^{\mu\nu}K^\sigma), \label{confKM}\\
	[K^\mu,P^\nu] &= - 2i (\eta^{\mu\nu}D + M^{\mu\nu}),\label{confKP}\\
	[K^\mu,K^\nu] &= 0, \label{confKK}
\end{align}
for $\mu,\nu,\sigma=0,1,2$. 

Equations (\ref{confKM}) and (\ref{confKK}) show, as is well known,
that the conformal algebra contains a second realisation of the
Poincaré algebra, given by the Lorentz generators and the
special-conformal ones.

One might wonder whether it is possible to repeat the construction of
the BMS-like algebras using the special-conformal generators instead
of the translation ones. A proposal in this direction was discussed in
\cite{Batlle:2020hia, Campello:2024th}, generalising the
special-conformal generators by making them dependent on an arbitrary
integer index, but the resulting algebra was in fact a W-algebra,
rather that a Lie algebra \cite{Fuentealba:2020zkf}.  In this context,
W-algebras are studied in Section \ref{sec:conformal-bms-w} from
the point of view of conformal field theory.

One can also study whether it is possible to write down some master
equation similar to (\ref{eigen2}) but for the special conformal
transformations. However, since in the canonical formalism these are
second order differential operators, one is led to consider the
commutator with the Lorentz Casimir instead of the action of the
Casimir on these operators. One can see that, if $K$ is any of the 3
special conformal generators in $2+1$, one has
\begin{equation}
  [C_2,K]=-2 K E,
  \label{C2Keq}
\end{equation}
with $E$ the Euler operator. It is possible to obtain other solutions
to this equation, in the form of second-order differential operators,
but in the end one finds out that the resulting set of generators
$K_n$, containing $K^0$ and the complex combinations of  $K^1$ and
$K^2$ for $n=0,\pm 1$, do not commute, except for $n=0,\pm 1$, and
hence do  not qualify as super-special conformal transformations.

A separate but related question is whether it is possible to consider
a generalised equation $[C_2,K]=\lambda K E$, with $\lambda$ arbitrary
and with the restriction of $K$ being a second-order differential
operator. A detailed study of the resulting set of differential
equations shows that there are solutions only for $\lambda=-2$ and
$\lambda=0$, and hence one does not have the liberty of adding the
parameter $\lambda$ that appears in the case of super-translations.

\section{The BRST complex of the Weyl $\lambda$-BMS algebra}
\label{sec:weyl-lambda-bms}

We will now depart from the mode algebra in equations
\eqref{WextLL}--\eqref{WextDD} and explore its possible central
extensions, resulting in the (centrally extended) Weyl $\lambda$-BMS
algebra.  We will then reformulate the centrally extended algebra in
terms of operator product expansions.  We will then discuss the
construction of Weyl $\lambda$-BMS strings by constructing the BRST
complex and showing that the BRST cohomology is isomorphic to the
chiral ring of a topologically twisted $N{=}2$ superconformal field
theory.

\subsection{The (centrally extended) Weyl $\lambda$-BMS algebra}
\label{sec:weyl-lambda-bms-voa}

Firstly, we reformulate the centreless Weyl $\lambda$-BMS algebra,
determine the possible central extensions and rewrite the centrally
extended algebra in terms of operator product expansions.

\subsubsection{The centreless Weyl $\lambda$-BMS algebra}
\label{sec:centr-weyl-lambda}

It is convenient to re-interpret the centreless Weyl $\lambda$-BMS
algebra in terms of natural objects associated to the punctured
complex plane.\footnote{The coordinates $(r,\phi)$ used for the
  lightcone in Section~\ref{sec:cano-real} are reminiscent of polar
  coordinates for the punctured complex plane.  There is however no
  change of variables relating $(r,\phi)$ to the complex coordinate
  $z$ used in this section.  One way to see this is to consider the
  functions annihilated by the Witt generators $L_n$.  For the
  realisation in this section, it is clear that any anti-holomorphic
  function -- i.e., any function of $\bar z$ -- is annihilated by
  $L_n = -z^{n+1} \partial$, whereas the operator $L_n$ given in
  equation~\eqref{superL} annihilates functions of $r e^{i n \phi}$,
  which clearly depends on $n$.}  To that end, we let
$A = \CC[z,z^{-1}]$ denote the associative algebra of Laurent
polynomials in a complex variable $z$.  The Lie algebra
$W = \Der _\CC A$ of derivations is the \textbf{Witt algebra}, which
is isomorphic to the Lie algebra of polynomial vector fields on the
circle.  Let $\lambda \in \ZZ$ and let $I(\lambda)$ denote the
one-dimensional $A$-module spanned by $(dz)^\lambda$; that is, a
typical vector in $I(\lambda)$ is of the form $f(z) (dz)^\lambda$ with
$f(z) \in A$.  The tacit understanding is that if $\lambda < 0$,
$(dz)^\lambda = \left( \frac{d}{dz} \right)^{-\lambda}$.  We introduce
bases $D_n := z^n$, $L_n := - z^{n+1} \partial$, where
$\partial f(z) = \frac{df}{dz}$ and
$P_n := z^{n-\lambda} (dz)^\lambda$, for $n \in \ZZ$, for the
$A$-modules $A$, $W$ and $I(\lambda)$, respectively.

We define a Lie algebra structure on the vector space $W \oplus A
\oplus I(\lambda)$ as follows:
\begin{itemize}
\item $W$ is a Lie algebra under the Lie bracket of vector fields:
  \begin{equation}
    [L_n,L_m] = (n-m) L_{n+m};
  \end{equation}
\item $A$ is an abelian Lie algebra, is acted on by $W$ as
  derivations:
  \begin{equation}
    [L_n, D_m] = -m D_{n+m};
  \end{equation}
\item $I(\lambda)$ is an abelian Lie algebra and is acted on by $W$ via the Lie
  derivative:
  \begin{equation}
    [L_n, P_m] = - (m+ \lambda n) P_{n+m}
  \end{equation}
  and by $A$ via the module action:
  \begin{equation}
    [D_m, P_n] = P_{m+n}.
  \end{equation}
  We could rescale $D_m$ by any nonzero $\mu$ and arrive at $[D_m,P_n]
  = \mu P_{m+n}$, without altering any of the other Lie brackets.  In
  particular, choosing $\mu = - \lambda$ we would arrive at the
  Lie bracket~\eqref{WextDP}.  As mentioned earlier, we will take $\mu
  = 1$ in what follows.
\end{itemize}

The Lie algebra generated by $\{L_n, D_n, P_n\}_{n\in\ZZ}$ given by
the above brackets is the double semi-direct product of Lie algebras
$\w_\lambda = (W\ltimes A)\ltimes I(\lambda)$ and it is
clearly isomorphic to the centreless Weyl $\lambda$-BMS algebra in
Section~\ref{sec:cano-real}: see
equations~\eqref{WextLL}--\eqref{WextDD}.  The Lie subalgebra
$W \ltimes A$ is often called the \textbf{Heisenberg--Virasoro
  algebra} \cite{MR1473758, MR946992, tan2021simple,
  adamovic2024nappiwittenvertexoperatoralgebra}, and corresponds to
the algebra of differential operators of degree at most $1$. For
$\lambda = -1$, $\w_{-1}$ agrees with the Weyl BMS Poisson algebra of
\cite{Batlle:2020hia}: see equations (4.16)-(4.21) in that paper.

\subsubsection{Central extensions}
\label{sec:central-extensions}

Our first order of business is to determine the possible central
extensions (up to equivalence).  They are classified by the second
cohomology of $\w_\lambda$ with values in the trivial
representation. We make use of the following lemma to simplify our
computation.

\begin{lemma}\label{lem:injective_map_in_cohomology}
  Let $\g$ be a Lie algebra and $\h\subset\g$ an ideal such that
  $[\g,\h] = \h$.  Then the canonical surjective homomorphism
  $\pi\colon\g\to\g/\h$ induces an injective linear map $\pi^*\colon
  H^2(\g/\h) \to H^2(\g)$ in cohomology.
\end{lemma}

\begin{proof}
  Recall\footnote{We refer the reader to \cite{MR874337} for a review of
    Lie algebra cohomology.} that the space of $p$-cochains of any Lie
  algebra $\g$ with values in the base field (here $\RR$) viewed as a
  trivial representation is given by
  $C^p(\g) = \Hom(\Lambda^p\g,\RR)$. We are particularly
  interested in the first few terms:
  \begin{equation}
    \begin{tikzcd}
      C^1(\g) \arrow[r,"d"] & C^2(\g) \arrow[r,"d"] & C^3(\g)
    \end{tikzcd}
  \end{equation}
  where for $\beta \in C^1(\g)$ and $\varphi \in C^2(\g)$, their
  differentials are given by
  \begin{equation}
    d\beta(X,Y) = -\beta([X,Y]) \qquad\text{and}\qquad d\varphi(X,Y,Z)
    = -\varphi([X,Y],Z) + \varphi([X,Z],Y) - \varphi([Y,Z],X).
  \end{equation}
  If $\h \subset \g$ is an ideal, the Chevalley--Eilenberg complex for
  the quotient Lie algebra $\g/\h$ can be understood as a subcomplex
  of $C^{\dotr}(\g)$.  Namely, any $\varphi \in C^p(\g/\h)$ may be
  extended uniquely to a $\varphi \in C^p(\g)$ such that
  $\varphi(X,\cdots) = 0$ for all $X \in \h$.  Such cochains are
  preserved by the Chevalley--Eilenberg differential precisely because
  $\h$ is an ideal and hence define a subcomplex.

  Now suppose that a cocycle $\varphi \in C^2(\g/\h)$ is a coboundary
  in $C^2(\g)$; that is, $\varphi = d\beta$ for some $\beta \in
  C^1(\g)$; that is, $\varphi(X,Y) = -\beta([X,Y])$.  We claim that
  $\beta \in C^1(\g/\h)$.  Indeed, suppose that $Z \in \h$.  By
  hypothesis, there exist $X_i \in \h$ and $Y_i \in \g$ such that $Z =
  \sum_i [X_i, Y_i]$.  Hence
  \begin{equation}
    \beta(Z) = \sum_i \beta([X_i,Y_i]) = - \sum_i \varphi(X_i,Y_i) = 0,
  \end{equation}
  since $X_i \in \h$ and $\varphi \in C^2(\g/\h)$.  In other words, if
  $\pi^*([\varphi]) = 0 \in H^2(\g)$, then $[\varphi] = 0 \in H^2(\g/\h)$.
  \end{proof}
  
\begin{remark}
  Lemma \ref{lem:injective_map_in_cohomology} requires $[\g,\h]=\h$,
  as can easily be seen when $\h$ is the centre of $\g$.
\end{remark}

\begin{proposition}
    The Lie algebra $\w_\lambda = (W\ltimes A)\ltimes
    I(\lambda)$ admits a 3-dimensional universal central extension
    generated by the 2-cocycles
    \begin{equation} \label{eq: Weyl lambda-BMS 2-cocycles}
        \begin{split}
        \gamma_{LL}(L_n, L_m) &= \frac{1}{12}n(n^2-1)\delta^0_{m+n}\\
        \gamma_{LD}(L_n, D_m) &= \frac{1}{2}n(n+1) \delta^0_{m+n}\\
        \gamma_{DD}(D_n, D_m) &= n\delta^0_{m+n}.
        \end{split}  
    \end{equation}
\end{proposition}

\begin{proof}
  Consider the Lie algebra $\g$ generated by $\{L_n, P_n, D_n,
  I_n\}_{n\in\ZZ}$, with (nonzero) Lie brackets
  \begin{equation}
    \label{eq: lambda-planar-GCA}
    \begin{alignedat}{2}
      &[L_n, L_m] = (n-m)L_{m+n}  \quad &&[L_n, D_m] = -m D_{m+n}\\
      &[L_n, P_m] = -(m+\lambda n)P_{m+n}  \quad &&[L_n, I_m] = (n-m)L_{m+n}\\
      &[D_m, P_n] = P_{m+n} \quad &&[D_m, I_n] = -I_{m+n}.
\end{alignedat}
\end{equation}
For $\lambda=-1$, this is known as the planar galilean conformal
algebra (GCA). It was shown by Gao, Liu and Pei that the second
cohomology group of the planar GCA (with values in the trivial module
$\CC$) is 3-dimensional \cite{MR3542533}. It is easy to check that
this statement holds true for any $\lambda\in\ZZ$.

Let $\h$ denote the abelian ideal in $\g$ spanned by $\{I_n\}_{n\in\ZZ}$.
Since $\w_\lambda=\g/\h$ and $[\g, \h] = \h$, we can make use of
Lemma~\ref{lem:injective_map_in_cohomology} to deduce that the
canonical surjective homomorphism $\tilde{\pi}\colon\g \to \w_\lambda$
induces an injective map in cohomology $\tilde{\pi}_*\colon H^2(\w_\lambda)\to
H^2(\g)$.  This implies that $\dim H^2(\w_\lambda) \leq 3$.

Now let $\h$ denote the abelian ideal of $\w_\lambda$ spanned by
$\{P_n\}_{n\in\ZZ}$. Indeed, $[\w_\lambda, \h]=\h$, and we have a
surjective Lie algebra homomorphism $\pi\colon\w_\lambda \to
\w_\lambda/\h=:\g_0$ (recall that $\g_\lambda$ is given by \cref{extLL,extLP,extPP}). 
Once again, using
Lemma~\ref{lem:injective_map_in_cohomology}, the injectivity of the
induced map $\pi^*\colon H^2(\g_0)\to H^2(\w_\lambda)$ implies that $\dim
H^2(\w_\lambda)\geq \dim H^2(\g_0)$. Since $\dim H^2(\g_0) =
3$ (as shown by Arbarello, De~Concini, Kac and Procesi \cite{MR946992}, see also \cite{MR2773310}), we arrive
at the nested inequality $3\leq\dim H^2(\w_\lambda)\leq3$. Thus,
$\dim H^2(\w_\lambda)=3$ indeed.  We may now obtain the
explicit form of the representative 2-cocycles by pulling back the
three representative 2-cocycles on $\g_0$ \cite{MR946992}
\begin{equation} \label{eq: g_0 2-cocycles}
  \begin{split}
    \gamma_{1}(L_n, L_m) &= \frac{1}{12}n(n^2-1)\delta^0_{m+n}\\
    \gamma_{2}(L_n, D_m) &= \frac{1}{2}n(n+1) \delta^0_{m+n}\\
    \gamma_{3}(D_n, D_m) &= n\delta^0_{m+n}
  \end{split}  
\end{equation}
by $\pi^*$ to get $\gamma_{LL} = \pi^*(\gamma_1)$, $\gamma_{LD}
=\pi^*(\gamma_2)$ and $\gamma_{DD}=\pi^*(\gamma_3)$ as given by
\eqref{eq: Weyl lambda-BMS 2-cocycles}.\footnote{Likewise, $\tilde{\gamma}_1 = \tilde{\pi}^* \circ
  \pi^*(\gamma_1)$, $\tilde{\gamma}_2 = \tilde{\pi}^* \circ
  \pi^*(\gamma_2)$ and $\tilde{\gamma}_3 = \tilde{\pi}^* \circ
  \pi^*(\gamma_3)$ agree with the expressions for the three representative
  two-cocycles on $\g$ \cite{MR3542533}.}
\end{proof}

\begin{remark}
  One might have expected that for the values of $\lambda$ for which
  the $\lambda$-BMS algebra admits additional central extensions
  (namely, $\lambda = -1,0,1$) so   would the Weyl $\lambda$-BMS
  algebra.  This however is not the case.  For example, for $\lambda =
  -1$, corresponding to the BMS algebra, the BMS algebra admits a
  central extension $c_P$ in the bracket $[L_m,P_n]$.  This however
  has to vanish in the Weyl BMS algebra: indeed, the zero mode $D_0$
  of the super-dilatations acts diagonally with $[D_0,L_m]$ = 0 and
  $[D_0,P_m]=P_m$ and hence by Jacobi $[D_0,[L_m,P_n]] = [L_m,P_n]$
  and hence there can be no central terms in $[L_m,P_n]$.
\end{remark}

\begin{definition}
  The universal central extension $\widehat{\w}_\lambda$ of
  $\w_\lambda$ is called the \textbf{Weyl $\lambda$-BMS
    algebra}.  For $\lambda = -1$, this is the three-dimensional
  version of the Weyl BMS algebra in \cite{Freidel:2021fxf}.
\end{definition}

\subsubsection{The Weyl $\lambda$-BMS algebra in terms of OPEs}
\label{sec:weyl-lambda-bms-opes}

Next we reformulate the Lie bracket of the Weyl $\lambda$-BMS
algebra in terms of operator product expansions for the fields
\begin{equation}
  T(z) = \sum_{n\in\ZZ} L_n z^{-n-2}, \qquad D(z) = \sum_{n\in\ZZ} D_n
  z^{-n-1} \qquad\text{and}\qquad P(z) = \sum_{n\in\ZZ} P_n z^{-n-(1-\lambda)}.
\end{equation}
The operator product algebra is given by
\begin{equation}
\label{eq: lambda-Weyl-BMS OPE}
  \begin{split}
    T(z) T(w) &= \frac{\tfrac12 c_L}{(z-w)^4} + \frac{2 T(w)}{(z-w)^2} + \frac{\partial T(w)}{z-w} + \reg\\
    T(z) D(w) &= \frac{c_{TD}}{(z-w)^3} + \frac{D(w)}{(z-w)^2} + \frac{\partial D(w)}{z-w} + \reg\\
    T(z) P(w) &= \frac{(1-\lambda)P(w)}{(z-w)^2} + \frac{\partial P(w)}{z-w} + \reg\\
    D(z) P(w) &= \frac{P(w)}{z-w} + \reg\\
    D(z) D(w) &= \frac{c_D}{(z-w)^2} + \reg\\
    P(z) P(w) &= \reg,
  \end{split}
\end{equation}
where $c_L$, $c_D$ and $c_{TD}$ are the three central charges. One can
check that the above operator product expansions are associative. This
and many of the calculations below have been performed using the
Mathematica package \texttt{OPEdefs} written by Kris Thielemans
\cite{Thielemans:1991uw,Thielemans:1992mu,Thielemans:1994er}.

For further reference, we remind the reader that the operator product
expansion can be reformulated in terms of a sequence of bilinear
products indexed by the integers:
\begin{equation}
  A(z) B(w) = \sum_{n\ll\infty} \frac{[A,B]_n(w)}{(z-w)^n},
\end{equation}
where the sum is over all $n$ less than some positive
integer.  The commutativity and associativity of the operator product
expansion translate into axioms for the brackets $[-,-]_n$ which are
reminiscent to those satisfied by the bracket in a Lie algebra.  All
the operator product algebras in this paper are conformal, so that
there is always a field $T(z)$ satisfying the operator product
expansion of the Virasoro algebra with some central charge (as in the
first operator product expansion in equation~\eqref{eq:
  lambda-Weyl-BMS OPE}).  We assume that all other fields have a
well-defined conformal weight, so that for any field $\Phi(z)$,
$[T,\Phi]_2 = h \Phi$, where $h$ is the conformal weight.  If $A(z)$
and $B(z)$ have conformal weights $h_A$ and $h_B$, respectively, we
expand them in modes according to
\begin{equation}
  A(z) = \sum_n A_n z^{-n-h_A} \qquad\text{and}\qquad B(z) = \sum_n
  B_n z^{-n-h_B}
\end{equation}
and the Lie algebra of modes can be read off from the singular part of
the operator product expansion via the formula
\begin{equation}
  [A_n,B_m] = \sum_{\ell \geq 1}
      \binom{n+h_A-1}{\ell -1} \left( [A,B]_\ell \right)_{m+n},
\end{equation}
with $\binom{n}{k} = \frac{n!}{k! (n-k)!}$ the usual binomial
coefficient.

\subsection{The BRST complex of the Weyl $\lambda$-BMS algebra}
\label{sec:brst-complex-Weyl-BMS}

We now construct the BRST complex for the Weyl $\lambda$-BMS algebra
above.  Its cohomology, which is the semi-infinite cohomology
of the $\lambda$-BMS algebra relative to the centre and with values in an
admissible representation, can be interpreted as the spectrum of (a
chiral sector of) a putative Weyl $\lambda$-BMS string, as was first
observed for the bosonic string in \cite{MR865483}.

\subsubsection{The ghosts BC systems}
\label{sec:ghosts-bc-systems}

The ghosts are described by fermionic BC systems $(b_i,c_i)$ for
$i=1,2,3$  with conformal weights $(2,-1)$, $(1,0)$ and
$(1-\lambda,\lambda)$, respectively.  Their operator product
expansions are the standard ones as in 
\begin{equation}
  \label{eq:BC-systems}
  b_i(z)c_j(w) = \frac{\delta_{ij}}{z-w} + \reg, \qquad b_i(z)b_j(w)
  = \reg \qquad\text{and}\qquad c_i(z) c_j(w) = \reg.
\end{equation}
The Virasoro element is given by
\begin{equation}
  \label{eq:Tgh}
  T^{\text{gh}} = -2 (b_1 \partial c_1) - (\partial b_1 c_1) - (b_2
  \partial c_2) - (1-\lambda) (b_3 \partial c_3) + \lambda (\partial b_3 c_3),
\end{equation}
where parentheses indicate the normal-ordered product, which
associates to the left so that $(ABC) := (A(BC))$, et cetera.  We
define the ghost number as usual by declaring $b_i$ to have
ghost number $-1$ and $c_i$ ghost number $+1$.

\subsubsection{The ghost Weyl $\lambda$-BMS algebra}
\label{sec:ghost-weyl-bms_l}

There is an embedding of the Weyl $\lambda$-BMS algebra in the
operator product algebra of the ghost BC systems.  The Virasoro
element is given by $T^{\text{gh}}$ in equation~\eqref{eq:Tgh}, which
results in a central charge
\begin{equation}
  \label{eq:cLgh}
  c^{\text{gh}}_L = -6 (5 - 2\lambda + 2 \lambda^2).
\end{equation}
One can find expressions for $D^{\text{gh}}$ and $P^{\text{gh}}$ after
some trial and error\footnote{Alternatively, one can follow the
  prescription given in \cite[Section~3]{Figueroa-OFarrill:2024wgs}
  for $\g_\lambda$ and apply it to $\w_\lambda$.} and this leads to
\begin{equation}
  \label{eq:Dgh}
  D^{\text{gh}} = (b_3 c_3) + \partial(c_1 b_2)
\end{equation}
which results in central charges
\begin{equation}
  c^{\text{gh}}_{TD} = 1- 2\lambda \qquad\text{and}\qquad
  {c_D}^{\text{gh}}= 1.
\end{equation}
Finally we find
\begin{equation}
  \label{eq:Mgh}
  P^{\text{gh}} = (c_1\partial b_3) + (c_2 b_3) + (1-\lambda) (\partial c_1 b_3).
\end{equation}
One finds that $T^{\text{gh}}$, $D^{\text{gh}}$ and $P^{\text{gh}}$
obey the Weyl $\lambda$-BMS algebra with the above values for the
central charges.  It should be mentioned that this is not the unique
embedding of the Weyl $\lambda$-BMS algebra in the operator product
algebra of the ghosts, but it is the one induced by the BRST
differential to be introduced presently.

\subsubsection{The BRST current}
\label{sec:brst-current}

The BRST current $J$ is a conformal weight $1$ and ghost number $1$
field of the form
\begin{equation}
  \label{eq:jBRST}
  J = (c_1 T) + (c_2 D) + (c_3 P)+ \cdots
\end{equation}
where $T$, $D$ and $P$ are a representation of the Weyl
$\lambda$-BMS algebra with opposite central charges:
\begin{equation}
  c_L = - c_L^{\text{gh}} = 6 (5 - 2\lambda + 2 \lambda^2), \qquad c_{TD}  = -c_{TD}^{\text{gh}} = 2\lambda -1 \qquad\text{and}\qquad c_D = -
  c_D^{\text{gh}}= -1.
\end{equation}
The fundamental property of $J$ is that its zero mode $d$ is a
differential: $d^2 = 0$, where the action of $d$ is given by the first
order pole of the operator product expansion with $J$; that is, $d =
[J,-]_1$.  By a result of Füsun Akman \cite{MR1250535}, $d$ is a
differential if and only if $T^{\text{tot}} := db_1$, $D^{\text{tot}}
:= db_2$  and $P^{\text{tot}} := db_3$ obey the centreless Weyl $\lambda$-BMS
algebra.  The BRST current is of course only defined up to the
addition of a total derivative, since that does not change its zero mode.

Some experimentation results in the following expression for the BRST
current:
\begin{equation}
  J = (c_1 T) + (c_2 D) + (c_3 P) + (b_1c_1\partial c_1) + (b_2 c_1 \partial c_2) + (b_3 c_1 \partial c_3) - \lambda (b_3 c_3 \partial c_1) + (c_2 b_3 c_3),
\end{equation}
which up to a total derivative takes the more usual form
\begin{equation}
\label{eq: BRST current semi-infinite REP}
  J' = (c_1 T) + (c_2 D) + (c_3 P) + \tfrac12 (c_1 T^{\text{gh}}) + \tfrac12
  (c_2 D^{\text{gh}}) + \tfrac12 (c_3 P^{\text{gh}}).
\end{equation}
One checks that $[J',J']_1$ is indeed a total derivative and that
$T^{\text{tot}} = [J',b_1]_1 = T + T^{\text{gh}}$, $D^{\text{tot}} =
[J',b_2]_1 = D + D^{\text{gh}}$  and $P^{\text{tot}} = [J',b_3]_1 = P
+ P^{\text{gh}}$ indeed give a representation of the centreless Weyl
$\lambda$-BMS algebra.

The form of the BRST current \eqref{eq: BRST current semi-infinite
  REP} indicates that the fields $T^{\text{gh}}$, $D^{\text{gh}}$ and
$P^{\text{gh}}$ are indeed generating functionals formed from the
semi-infinite wedge (i.e., fermionic Fock) representation of the Weyl
$\lambda$-BMS algebra.

\subsection{Quasi-isomorphism with a twisted $N{=}2$ superconformal algebra}
\label{sec:n=2SCA}

The BRST cohomology is not just a graded vector space, but admits a
richer algebraic structure, first formalised in the context of the
bosonic string in \cite{Lian:1992mn} based on initial observations of
the BRST cohomology of noncritical bosonic strings in
\cite{Witten:1992yj}.  The relevant structure is that of a
Batalin--Vilkovisky (BV) algebra: a special kind of Gerstenhaber algebra
where the bracket measures the failure of a second order operator
(here the Virasoro antighost zero mode) being a derivation over the
normal ordered product.  We refer to \cite{Lian:1992mn} for the
relevant facts and definitions.

A paradigmatic example of BV algebra is provided by the chiral ring of a
topologically twisted $N=2$ superconformal algebra
\cite{Witten:1988ze,Eguchi:1990vz}.  The twisting gives the two
supercharges $G^\pm$, originally of conformal weight $\frac32$,
conformal weights $1$ and $2$.  The zero mode of the supercharge with
conformal weight $1$ plays the rôle of the BRST differential and the
zero mode of the supercharge with conformal weight $2$ plays the rôle
of the Virasoro antighost zero mode.

As the case of the bosonic string theory shows, not every topological
conformal field theory is of this form.  Indeed, this is typical of
most string theories, with the exception of the $N=2$ string itself
\cite{Gomis:1991ue,Giveon:1993ew}.  Nevertheless, exploiting the
embedding \cite{Berkovits:1993xq} of the $N=1$ NSR string into the
$N=2$ string, it was shown in \cite{Marcus:1994nd}, that the BRST
cohomology of the $N=1$ string is isomorphic to the chiral ring of an
$N=2$ superconformal field theory.  This then suggested how to prove
the same result for other string theories
\cite{Figueroa-OFarrill:1995qkv},
resulting in a natural conjecture that the BRST cohomology of any string theory and, more
generally, of any two-dimensional topological conformal field theory
is isomorphic to the chiral ring of some twisted $N=2$ superconformal
field theory
\cite{Figueroa-OFarrill:1995agp,Figueroa-OFarrill:1996dic}.

In this section we will test this conjecture (and show that it holds)
for the BRST cohomology of the Weyl $\lambda$-BMS algebra.

To do this we will embed into the tensor product of the BRST complex
of the Weyl $\lambda$-BMS algebra with a Koszul topological conformal
algebra a topologically twisted $N{=}2$ superconformal algebra.  This
embedding will turn out to be quasi-isomorphic, in that the BRST
cohomology of the Weyl $\lambda$-BMS algebra will be isomorphic (as a
BV algebra) to the chiral ring of the $N{=}2$ superconformal algebra.
This gives further evidence to the conjecture in
\cite{Figueroa-OFarrill:1995qkv,Figueroa-OFarrill:1995agp,Figueroa-OFarrill:1996dic}.

We start by modifying the BRST current by a total derivative in order
to simplify its operator product.  Let us define
\begin{equation}
  \begin{split}
    \GG^+_{\text{W}}&:= J + \partial \left( (c_1 b_2 c_2) +  (1 + \lambda) (c_1 b_3 c_3) + c_2 + \tfrac12 (7-4\lambda) \partial c_1\right)\\
    \GG^-_{\text{W}} &:= b_1\\
    \JJ_{\text{W}} &:= - (b_1c_1) - (b_2 c_2) - (b_3 c_3)\\
    \TT_{\text{W}} &:= T^{\text{tot}}.
  \end{split}
\end{equation}
The resulting operator product expansions are
\begin{equation}
  \begin{split}
    \GG^+_{\text{W}}(z) \GG^+_{\text{W}}(w) &= \frac{(\lambda-1) \partial(c_1\partial^2 c_1)(w)}{z-w} + \reg\\
    \GG^+_{\text{W}}(z) \GG^-_{\text{W}}(w) &= \frac{(7-4\lambda)}{(z-w)^3} + \frac{\JJ_{\text{W}}(w)}{(z-w)^2} + \frac{\TT_{\text{W}}(w)}{z-w} + \reg\\
    \GG^-_{\text{W}}(z) \GG^-_{\text{W}}(w) &= \reg\\
    \TT_{\text{W}}(z) \JJ_{\text{W}}(w) &=\frac{(2\lambda-5)}{(z-w)^3} + \frac{\JJ_{\text{W}}(w)}{(z-w)^2} + \frac{\partial\JJ_{\text{W}}(w)}{z-w} + \reg\\
    \TT_{\text{W}}(z) \GG^+_{\text{W}}(w)&= \frac{6(\lambda-1)c_1(w)}{(z-w)^4} + \frac{2(\lambda-1)\partial c_1(w)}{(z-w)^3} + \frac{\GG^+_{\text{W}}(w)}{(z-w)^2} + \frac{\partial\GG^+_{\text{W}}(w)}{z-w} + \reg\\
    \TT_{\text{W}}(z) \GG^-_{\text{W}}(w) &= \frac{2 \GG^-_{\text{W}}(w)}{(z-w)^2} + \frac{\partial\GG^-_{\text{W}}(w)}{z-w} + \reg\\
    \JJ_{\text{W}}(z) \JJ_{\text{W}}(w) &= \frac{3}{(z-w)^2}\\
    \JJ_{\text{W}}(z) \GG^+_{\text{W}}(w) &= \frac{6(1-\lambda)c_1(w)}{(z-w)^3} + \frac{4(1-\lambda)\partial c_1(w)}{(z-w)^2} + \frac{\GG^+_{\text{W}}(w)}{z-w} + \reg\\
    \JJ_{\text{W}}(z) \GG^-_{\text{W}}(w) &= \frac{-\GG^-_{\text{W}}(w)}{z-w} + \reg,
  \end{split}
\end{equation}
which are clearly not those of a topologically twisted $N=2$
superconformal algebra.  For example, $\GG^+_{\text{W}}$ does not have
regular operator product expansion with itself.

\subsubsection{The Koszul topological conformal algebra}
\label{sec:kosz-topol-conf}

We now tensor with a Koszul topological conformal algebra consisting
of one fermionic BC system $(b,c)$ and one bosonic BC system
$(\beta,\gamma)$, both of conformal weights $(1-\mu,\mu)$, with basic
operator product expansions:
\begin{equation}
  b(z) c(w) = \frac{1}{z-w} + \reg \qquad\text{and}\qquad \beta(z)
  \gamma(w) = \frac{1}{z-w} + \reg.
\end{equation}
The Koszul topological conformal algebra embeds a twisted $N{=}2$
superconformal algebra given by the following fields:
\begin{equation}\label{eq:koszul-neq2}
  \begin{split}
    \GG^+_{\text{K}}&:= (b \gamma)\\
    \GG^-_{\text{K}} &:= (1-\mu) (\partial c \beta) - \mu (c \partial\beta)\\
    \JJ_{\text{K}} &:= \mu (b c) + (1-\mu) (\beta\gamma)\\
    \TT_{\text{K}} &:= (1-\mu) (\beta\partial\gamma) - \mu (\partial\beta\gamma) -
    (1-\mu) (b \partial c) + \mu (\partial b c).
  \end{split}
\end{equation}
These fields obey a twisted $N{=}2$ superconformal algebra on the nose:
\begin{equation}
  \begin{split}
    \GG^\pm_{\text{K}}(z) \GG^\pm_{\text{K}}(w) &= \reg\\
    \GG^+_{\text{K}}(z) \GG^-_{\text{K}}(w) &= \frac{(2\mu-1)}{(z-w)^3} + \frac{\JJ_{\text{K}}(w)}{(z-w)^2} + \frac{\TT_{\text{K}}(w)}{z-w} + \reg\\
    \JJ_{\text{K}}(z) \GG^\pm_{\text{K}}(w) &= \frac{\pm\GG^\pm_{\text{K}}(w)}{z-w} + \reg,
  \end{split}
\end{equation}
from which the remaining operator product expansions follow by associativity, as shown independently in \cite{Figueroa-OFarrill:1993lqr} and \cite{Getzler:1993py}:
\begin{equation}
  \begin{split}
    \JJ_{\text{K}}(z) \JJ_{\text{K}}(w) &= \frac{(2\mu-1)}{(z-w)^2}\\
    \TT_{\text{K}}(z) \JJ_{\text{K}}(w) &=\frac{(1-2\mu)}{(z-w)^3} + \frac{\JJ_{\text{K}}(w)}{(z-w)^2} + \frac{\partial\JJ_{\text{K}}(w)}{z-w} + \reg\\
    \TT_{\text{K}}(z) \GG^+_{\text{K}}(w)&= \frac{\GG^+_{\text{K}}(w)}{(z-w)^2} + \frac{\partial\GG^+_{\text{K}}(w)}{z-w} + \reg\\
    \TT_{\text{K}}(z) \GG^-_{\text{K}}(w) &= \frac{2 \GG^-_{\text{K}}(w)}{(z-w)^2} + \frac{\partial\GG^-_{\text{K}}(w)}{z-w} + \reg\\
    \TT_{\text{K}}(z) \TT_{\text{K}}(w) &= \frac{2 \TT_{\text{K}}(w)}{(z-w)^2} + \frac{\partial \TT_{\text{K}}(w)}{z-w} + \reg.
  \end{split}
\end{equation}

\subsubsection{A twisted $N{=}2$ superconformal algebra}
\label{sec:twist-n=2-superc}

Now let
\begin{equation}
  X := (1-\lambda ) \left( (\partial c_1 c_1 c\beta) + (c_1 \beta\gamma) -
    (c_1 b c) - \partial c_1 \right)
\end{equation}
and define the following fields
\begin{equation}
  \begin{split}
    \GG^+ &:= \GG^+_{\text{W}} + \GG^+_{\text{K}} + \partial X\\
    \GG^- &:= \GG^-_{\text{W}} + \GG^-_{\text{K}}.
  \end{split}
\end{equation}
It follows by calculation that
\begin{equation}
  \begin{split}
    \GG^\pm(z) \GG^\pm(w) &= \reg\\
    \GG^+(z) \GG^-(w) &= \frac{2 (2+\mu-\lambda)}{(z-w)^3} +
    \frac{\JJ(w)}{(z-w)^2} + \frac{\TT(w)}{z-w} + \reg,
  \end{split}
\end{equation}
which defines $\JJ$ and $\TT$:
\begin{equation}
  \begin{split}
    \JJ &= \JJ_{\text{W}} + \JJ_{\text{K}} + (1 - \lambda) \left( (b
      c) - (\beta\gamma) + \partial (c_1 c \beta)  \right)\\
    \TT &= \TT_{\text{W}} + \TT_{\text{K}}.
  \end{split}
\end{equation}
Another calculation shows that
\begin{equation}
  \JJ(z) \GG^\pm(w) = \frac{\pm\GG^\pm(w)}{z-w} + \reg,
\end{equation}
from which the other operator product expansions of the twisted
$N{=}2$ superconformal algebra follow by associativity:
\begin{equation}
  \begin{split}
    \JJ(z) \JJ(w) &= \frac{2(2+\mu-\lambda)}{(z-w)^2} + \reg \\
    \TT(z) \JJ(w) &= \frac{-2(2+\mu-\lambda)}{(z-w)^3} + \frac{\JJ(w)}{(z-w)^2} + \frac{\partial\JJ(w)}{z-w} + \reg \\
    \TT(z) \GG^+(w) &= \frac{\GG^+(w)}{(z-w)^2} + \frac{\partial\GG^+(w)}{z-w} + \reg \\
    \TT(z) \GG^-(w) &= \frac{2\GG^-(w)}{(z-w)^2} + \frac{\partial\GG^-(w)}{z-w} + \reg \\
    \TT(z) \TT(w) &= \frac{2\TT(w)}{(z-w)^2} + \frac{\partial\TT(w)}{z-w} + \reg.
  \end{split}
\end{equation}
Notice that by choosing $\mu$ appropriately, we can bring the central
charges to any desired values, e.g., if $\mu = \lambda - 2$, then all
central terms vanish.  Notice that since $\GG^+$ differs from
$\GG^+_{\text{W}} + \GG^+_{\text{K}}$ by a total derivative, the
$N{=}2$ differential $d_{N{=}2}=[\GG^+,-]_1$ is the sum of the Weyl
$\lambda$-BMS and Koszul differentials and hence we may apply the
Künneth theorem to deduce that the chiral ring of the $N{=}2$
superconformal algebra is isomorphic to the graded tensor product of
the cohomology of the Weyl $\lambda$-BMS differential $d$ and the
cohomology of the Koszul differential
$d_{\text{K}} = [\GG^+_{\text{K}},-]_1$.  Since (once we choose a
picture for the $\beta\gamma$ system) the latter cohomology is trivial
except in degree $0$ and isomorphic to $\CC$ there, we obtain that the
chiral ring is isomorphic to the BRST cohomology of the Weyl
$\lambda$-BMS algebra.

It is also the case that the isomorphism is one of BV algebras.
Indeed, that the BRST cohomology of the Weyl $\lambda$-BMS algebra
admits the structure of a BV algebra follows from results in
\cite{Figueroa-OFarrill:1995qkv,MR1466615}, which guarantee this is
the case simply because the Virasoro antighost is a conformal primary
with weight $2$.  In the BRST cohomology of the Weyl $\lambda$-BMS
algebra, the BV differential is given by the zero mode of the Virasoro
antighost ($\GG^-_{\text{W}}$), whereas in the case of the
topologically $N=2$ superconformal algebra it is given by the zero
mode of $\GG^- = \GG^-_{\text{W}} + \GG^-_{\text{K}}$, which acts the
same way in cohomology, since $\GG^-_{\text{K}}$ acts trivially on the
Kozul cohomology.

\section{The conformal BMS W-algebra}
\label{sec:conformal-bms-w}

Consider again the Weyl BMS algebra, which is the Weyl $\lambda$-BMS
algebra with $\lambda = -1$.  There is no way to extend it to a Lie
algebra by the addition of a field $K(z)$ of conformal weight $2$ in
such a way that the operator expansion $K(z)P(w)$ is nonzero, but as
shown in \cite{Fuentealba:2020zkf}, such an extension exists as a
W-algebra.  In this section we show that contrary to many of the
W-algebras which have been studied in this context, this one does not
admit a BRST complex and hence there is no natural notion of W-strings
for it.

\subsection{The W-algebra}
\label{sec:w-algebra}

We consider the VOA generated by fields $T, D, K, P$ with the
following operator product expansions:
\begin{equation}
  \label{eq:conf-BMS-W-OPEs}
  \begin{split}
    T(z) T(w) &= \frac{\tfrac12 c_L}{(z-w)^4} + \frac{2 T(w)}{(z-w)^2} + \frac{\partial T(w)}{z-w} + \reg\\
    T(z) D(w) &= \frac{D(w)}{(z-w)^2} + \frac{\partial D(w)}{z-w} + \reg\\
    T(z) P(w) &= \frac{2 P(w)}{(z-w)^2} + \frac{\partial P(w)}{z-w} + \reg\\
    T(z) K(w) &= \frac{2 K(w)}{(z-w)^2} + \frac{\partial K(w)}{z-w} + \reg\\
    D(z) D(w) &= \frac{c_D}{(z-w)^2} + \reg\\
    D(z) K(w) &= -\frac{K(w)}{z-w} + \reg\\
    D(z) P(w) &= \frac{P(w)}{z-w} + \reg\\
    K(z) K(w) &= \reg\\
    P(z) P(w) &= \reg,
  \end{split}
\end{equation}
and the operator product expansion $K(z) P(w)$ has the following
singular terms
\begin{equation}
  K(z)P(w) = \sum_{\ell=1}^4 \frac{[K,P]_\ell(w)}{(z-w)^\ell}
\end{equation}
with
\begin{equation}
  \begin{split}
    [K,P]_4 &= \frac{3(c_D-1) c_D^2}{1+c_D}\\
    [K,P]_3 &= \frac{3(1-c_D)c_D}{1+c_D} D\\
    [K,P]_2 &= -c_D T + \frac{3(1-c_D)c_D}{2(1+c_D)} \partial D + \frac{2 c_D-1}{1+c_D} D^2\\
    [K,P]_1 &= -\tfrac12 c_D \partial T - \frac{1+c_D^2}{2(1+c_D)} \partial^2 D + \frac{2c_D -1}{1+c_D} D\partial D + TD - \frac{1}{1+c_D} D^3.
  \end{split}
\end{equation}
Associativity of the operator product forbids any other central
charges and even forces the Virasoro central charge to be
\begin{equation}
  c_L = \frac{-2(6 c_D^2 - 8 c_D +1)}{1+c_D}.
\end{equation}

The structure of the above operator product algebra can be understood
a bit better if we define the following primary fields:
\begin{equation}
  \begin{split}
    \Phi_2 &= \frac{2c_D - 1}{1+c_D} \left( D^2 - \frac{2c_D}{c_L} T  \right) \\
    \Phi_3 &= -\frac{1}{1+c_D} D^3 + \frac{1}{3-2c_D} \left( TD - \tfrac12 \partial^2 D \right),
  \end{split}
\end{equation}
of conformal weights $2$ and $3$, respectively.  Then the operator
product expansion $K(z)P(w)$ simply contains the conformal families of
the identity, $D$, $\Phi_2$ and $\Phi_3$ with coefficients
\begin{equation}
  K(z) P(w)= \frac{3 c_D^ 2(c_D-1)}{1 + c_D} \frac{[\mathbb{1}](w)}{(z-w)^4} + \frac{3 c_D(1-c_D)}{1+c_D} \frac{[D](w)}{(z-w)^3} + \frac{[\Phi_2](w)}{(z-w)^2} + \frac{[\Phi_3](w)}{z-w} + \reg,
\end{equation}
where for a primary field $\phi$, the notation $[\phi]$ stands for its
conformal family.  Explicitly, in the above equation and up to regular
terms, we have
\begin{equation}
  \begin{split}
    [\mathbb{1}](w) &= \mathbb{1} + \frac{4}{c_L} (z-w)^2 T(w) + \frac{2}{c_L}(z-w)^3 \partial T(w)\\
    [D](w) &= D(w) + \tfrac12 (z-w) \partial D(w) + \frac{2(1+c_D)}{3 c_D(3-2c_D)} (z-w)^2 \left( (TD)(w) -\frac{2c_D^2 - c_D + 2}{4(1+c_D)} \partial^2 D(w)\right)\\
    [\Phi_2](w) &= \Phi_2(w) + \tfrac12 (z-w) \partial \Phi_2(w)\\
    [\Phi_3](w) &= \Phi_3(w).
  \end{split}
\end{equation}

\subsection{Non-existence of BRST complex}
\label{sec:non-existence-brst}

The proof of non-existence of the BRST complex for the above W-algebra
is computational, but we will give some details setting up the
calculation and then explain the result.

We have four quasiprimary fields in the W-algebra: $T,D,P,K$ of
conformal weights $2,1,2,2$, respectively.  We introduce fermionic
ghost systems $(b_i,c_i)$ for $i=1,2,3,4$ of weights $(2,-1)$ for
$i=1,3,4$ and weights $(1,0)$ for $i=2$.  As usual we assign ghost
numbers $1$ to the $c_i$ and $-1$ to the $b_i$.  The putative BRST
current has ghost number $1$ and conformal weight $1$ and takes the
form
\begin{equation}
  J = c_1 T + c_2 D + c_3 P + c_4 K + \cdots
\end{equation}
where $\cdots$ refers to any terms of with one or more antighosts
$b_i$.  Our methodology is naive.  We write the most general $J$ of
the above form and of conformal weight $1$ and ghost number $1$ and
demand that $d^2 = 0$, with $d := [J,-]_1$ its zero mode.   The
calculations have been performed in Mathematica on a 2020 MacBook Pro
laptop with a 2.3 GHz Quad-Core Intel Core i7 processor and 16Gb of
RAM, using the package \texttt{OPEdefs} (version 3.1 beta 4) written
by Kris Thielemans
\cite{Thielemans:1991uw,Thielemans:1992mu,Thielemans:1994er}.  A
notebook is available upon request.

Table~\ref{tab:ingr-bcc-curr} lists the ingredients out of which we
may write the terms in $J$ of the form $BC^2$, along with their
conformal weights.  This and the following table is to be supplemented
with the following  table of the fields of low conformal weight made
out of the generators of the original W-algebra:
\begin{equation*}
  \begin{tabular}{>{$}c<{$}|>{$}l<{$}}
    \omega_X & \text{fields}\\\toprule
    0 & \1\\
    1 & D\\
    2 & T,P,K,D^2,\partial D.
    \end{tabular}
\end{equation*}
All the $BC^2$ terms are obtained by picking one term from each table
and ensuring that the sum of the conformal weights
$\omega_B + \omega_{C^2} + \omega_X = 1$. It is easy to see that the
possible triples $(\omega_B, \omega_{C^2}, \omega_X)$ of conformal
weights are $(1,-2,2)_{15}$, $(1,-1,1)_{12}$, $(1,0,0)_{18}$,
$(2,-2,1)_{12}$, $(2,-1,0)_{48}$ and $(3,-2,0)_{12}$, where the
subscript is the multiplicity.  This means there are $117$ such terms,
which despite being easy to enumerate, we will refrain from doing so
here.

\begin{table}[h!]
\caption{Ingredients of $BC^2$ terms with their conformal weights}
  \label{tab:ingr-bcc-curr}
  \begin{minipage}[t]{0.2\linewidth}
    \centering
    \begin{tabular}{>{$}c<{$}|>{$}l<{$}}
      \omega_B & \text{fields}\\\toprule
      1 & b_2\\
      2 & b_1,\partial b_2, b_3, b_4\\
      3 & \partial b_1,\partial^2 b_2, \partial b_3, \partial b_4\\
    \end{tabular}
  \end{minipage}
  \begin{minipage}[t]{0.75\linewidth}
    \centering
    \begin{tabular}{>{$}c<{$}|>{$}l<{$}}
      \omega_{C^2} & \text{fields}\\\toprule 
    -2 & c_1 c_3, c_1 c_4, c_3 c_4\\
    -1 & c_1 \dc_3, \dc_1 c_3, c_1 \dc_4, \dc_1 c_4, c_3 \dc_4, \dc_3 c_4,\\
       & c_1 \dc_1, c_3 \dc_3, c_4 \dc_4, c_1 c_2, c_2 c_3, c_2 c_4\\
     0 & c_1 \ddc_3, \dc_1 \dc_3, \ddc_1 c_3, c_1 \ddc_4, \dc_1 \dc_4,\\
       & \ddc_1 c_4, c_3 \ddc_4, \dc_3 \dc_4, \ddc_3 c_4, c_1 \ddc_1\\
       & c_3 \ddc_3, c_4\ddc_4, \dc_1 c_2, c_1 \dc_2, \dc_2 c_3,\\
       & c_2 \dc_3, \dc_2 c_4, c_2 \dc_4\\\bottomrule
    \end{tabular}
  \end{minipage}
\end{table}

Table~\ref{tab:ingr-bbccc-curr} lists the ingredients in terms of the
form $B^2C^3$ along with their conformal weights.  Again all terms are
obtained by picking one term from each table and ensuring that the sum
of the conformal weights $\omega_{B^2} + \omega_{C^3} + \omega_X = 1$.
It is again easy to see that the possible triples
$(\omega_{B^2}, \omega_{C^3}, \omega_X)$ of conformal weights are
$(4,-3,0)_{10}$, $(3,-3,1)_4$ and $(3,-2,0)_{48}$ for a total of $62$,
which we will also refrain from listing.  There are no terms with
three (or more) antighosts because the conformal weight of any term of
the form $B^3C^4X$ is bounded below by $2$ and this only increases
with terms with higher number of antighosts.  In total there are $179$
possible terms of ghost number 1 and conformal weight 1 we could add
to the BRST current.

\begin{table}[h!]
\caption{Ingredients of $B^2C^3$ terms with their conformal weights}
\label{tab:ingr-bbccc-curr}
  \begin{minipage}[t]{0.2\linewidth}
    \centering
    \begin{tabular}{>{$}c<{$}|>{$}l<{$}}
      \omega_{B^2} & \text{fields}\\\toprule 
      3 & b_1 b_2, b_2 b_3, b_2 b_4, b_2 \db_2\\
      4 & b_1 b_3, b_1 b_4, b_3 b_4, b_2 \ddb_2,\\
        & \db_1 b_2, b_1 \db_2, \db_2 b_3,\\
        & b_2 \db_3, \db_2 b_4, b_2 \db_4 \\\bottomrule
    \end{tabular}
  \end{minipage}
  \begin{minipage}[t]{0.75\linewidth}
    \centering
    \begin{tabular}{>{$}c<{$}|>{$}l<{$}}
      \omega_{C^3} & \text{fields}\\\toprule 
    -3 & c_1 c_3 c_4\\
    -2 & c_1 c_2 c_3, c_1 c_2 c_4, c_2 c_3 c_4, c_1 \dc_1 c_3,\\
       &  c_1 \dc_1 c_4, c_1 c_3 \dc_3, c_1 c_4 \dc_4, c_3 \dc_3 c_4,\\
       & c_3 c_4 \dc_4,\dc_1 c_3 c_4, c_1 \dc_3 c_4, c_1 c_3 \dc_4\\
      \bottomrule
    \end{tabular}
  \end{minipage}
\end{table}

Given the most general $J$, depending on $179$ parameters, we
calculate the first-order pole $[J,J]_1$ in the operator product
expansion of the putative BRST current with itself.  The equations are
then $[[J,J]_1,\phi]_1=0$ for $\phi$ one of the generating fields of the
VOA: $b_i$, $c_i$, $T$, $D$, $P$ and $K$.  This results in $3288$
equations, which admit no solutions.  The basic reason is the
following.  There are two central charges in the W-algebra: the
Virasoro central charge $c_L$ and that of the field $D$, denoted
$c_D$.  The existence of the BRST complex requires both of them
to be critical, by which we mean that the values of $c_L$ and $c_D$
should cancel the ones of the ghost representation.  However these
central charges are not independent: associativity of the operator
product expansion relates them:
\begin{equation}
  c_L = \frac{-2 (1 - 8 c_D + 6 c_D^2)}{1 + c_D}.
\end{equation}
The critical values are $c_L= 80$ and $c_D = -2$, which do not satisfy
the above equation.  Indeed, when $c_D=-2$, one finds that $c_L=82$.
(It is intriguing that the excess is small and integral.)

This result is perhaps surprising given our experience with other
W-algebras.  Deformable W-algebras, those which exist for generic
values of the Virasoro central charge $c_L$, are typically constructed
via Drinfel'd--Sokolov reduction \cite{MR0760998,deBoer:1993iz}.  The
starting point of this reduction is the affine Kac--Moody algebra
associated to the vacuum-preserving subalgebra $\g$ (a contraction
$c_L\to \infty$ of the algebra of the vacuum-preserving modes)
\cite{Bowcock:1991zk}.  The vacuum-preserving Virasoro modes
$L_{\pm 1},L_0$ define an
$\mathfrak{sl}(2,\RR) \cong \mathfrak{so}(2,1)$ subalgebra of $\g$ and
these are in bijective correspondence with adjoint orbits of nilpotent
elements in $\g$, at least for $\g$ semisimple, which is typically the
case.  The best known W-algebras are associated to the principal
nilpotent orbits, those with the smallest stabiliser, namely
$\mathfrak{so}(2,1)$ itself.  Perhaps the best known such example is
the W${}_3$ algebra \cite{Zamolodchikov:1985wn}, whose BRST
differential was constructed in \cite{Thierry-Mieg:1987crv} and whose
cohomology was studied in detail in \cite{Bouwknegt:1995vx}.  Those
W-algebras always admit a BRST complex and indeed a reasonable notion
of semi-infinite cohomology \cite{MR4215746}, as is the case for Lie
algebras \cite{MR0740035}.  We also have constructions of BRST
complexes for W-algebras associated with minimal nilpotent orbits, as
in \cite{Bershadsky:1990bg}, which constructs a W-algebra out of the
minimal nilpotent orbit of $\mathfrak{sl}(3,\RR)$; although we are not
aware of any general result for the existence of a semi-infinite
cohomology theory for the minimal orbits.  The conformal BMS W-algebra
is one of four W-algebras which can be obtained by Drinfel'd--Sokolov
reduction of the three-dimensional conformal algebra
$\mathfrak{so}(3,2)$ \cite{Gupta:2023fmp}, but it is not the W-algebra
associated to the principal nilpotent orbit (whose BRST differential
was constructed in \cite{Zhu:1993mb}), nor indeed to the minimal
nilpotent orbit (with stabiliser
$\mathfrak{so}(2,2) \cong \mathfrak{so}(2,1) \oplus
\mathfrak{so}(2,1)$), but to an intermediate nilpotent orbit with
stabiliser $\mathfrak{so}(2,1) \oplus \mathfrak{so}(1,1)$.  (There is
another intermediate orbit with stabiliser
$\mathfrak{so}(2,1) \oplus \mathfrak{so}(2)$ which can also be shown
to lack a BRST complex.)  There are, to our knowledge, no theorems
about the existence of a semi-infinite differential for such
W-algebras and this example suggests that perhaps we should not expect
them to exist.

\section{Conclusions and outlook}
\label{sec:conclusions-outlook}

We have constructed the Weyl $\lambda$-BMS algebra in three
dimensions: an extension of the three-dimensional Lorentz algebra
$\so(2,1)$ by super-translations, super-rotations and super-dilatations,
which agrees for $\lambda =-1$ with the Weyl--BMS algebra in the
literature.  We construct this algebra out of the Fourier modes
of a free massless Klein--Gordon field, a construction in which the
quadratic Casimir of $\so(2,1)$ plays a crucial rôle.  We then
reformulate the Weyl $\lambda$-BMS algebra in terms of operator
product expansions and show that it admits a three-parameter family of
central extensions.  We construct the BRST complex for critical values
of the three central charges and show that the BRST cohomology is
isomorphic to the chiral ring of a topologically twisted $N=2$
superconformal field theory.  Returning to the classical case of
$\lambda=-1$, we argue that there is no ``conformal'' BMS Lie algebra
obtained by further extending the Weyl--BMS Lie algebra by super
special-conformal transformations.  There exists a classical conformal
BMS W-algebra which we fully quantise in the language of operator
product expansions and argue that the resulting W-algebra does not
admit a BRST complex.

Some of these results can be extended to define a Weyl $\lambda$-BMS
Lie superalgebra and a fully quantum conformal BMS W-superalgebra and
we will report on this in future work.

An interesting question we do not know the answer to is whether the
conformal BMS W-algebra admits a canonical realisation.

Another interesting question is whether there exist string sigma
models with gauge algebra given by the Weyl--BMS Lie algebra.  It is
worth remarking that the flat ambitwistor string \cite{MasonSkinner}
gives a realisation of the Weyl--BMS Lie algebra.  Indeed, if we let
$(X^\mu, \Pi_\mu)$ be $d+1$ bosonic BC systems with conformal weights
$(0,1)$ which describe the flat ambitwistor string, then
\begin{equation}
  T = - \partial X^\mu \Pi_\mu, \qquad D = \tfrac12 X^\mu \Pi_\mu
  \qquad\text{and}\qquad P = \tfrac12 \eta^{\mu\nu} \Pi_\mu \Pi_\nu
\end{equation}
provide a realisation of the Weyl--BMS Lie algebra with central
charges $c_L = 2(d+1)$, $c_D = -\tfrac14 (d+1)$ and $c_{TD} =
-\tfrac12(d+1)$.  The critical values of the central charges of a
Weyl--BMS string are given by $c_L = 54$, $c_D = -\tfrac{27}{4}$ and
$c_{TD} = -\tfrac{27}{2}$.  Even if we were to rescale $D$, we would
not find the critical values for a putative Weyl--BMS string.

\section*{Acknowledgments}
\label{sec:acknowledgments}

CB and JG would like to thank Víctor Campello for useful discussions.
JMF and GSV are grateful to Lucas Buzaglo and Sue Sierra for sharing
their expertise about the Witt algebra and their modules and in
particular for pointing out \cite{MR2320470}.  JG acknowledges the
hospitality of the Faculty of Sciences and the Institute of Exact and
Natural Sciences of the Universidad Arturo Prat, where this work was
completed. The stay was partially funded by the SIA 85220027.  We
would like to thank Ricardo Troncoso for making us aware of
\cite{Grumiller:2019fmp}, and Rudranil Basu for pointing out Appendix
A of \cite{Banerjee:2020qjj} to us.

The work of CB is partially supported by Project MAFALDA
(PID2021-126001OBC31, funded by MCIN/ AEI /10.13039/50110001 1033 and
by ``ERDF A way of making Europe''), Project MASHED
(TED2021-129927B-I00, funded by MCIN/AEI /10.13039/501100011033 and by
the ``European Union Next GenerationEU/PRTR''), and Project ACaPE
(20121-SGR-00376, funded by AGAUR-Generalitat de Catalunya). The
research of JG was supported in part by PID2022-136224NB-C21 and by
the State Agency for Research of the Spanish Ministry of Science and
Innovation through the Unit of Excellence María de Maeztu 2020-2023
award to the Institute of Cosmos Sciences (CEX2019- 000918-M).
GSV is supported by a studentship from the UK Science and Technologies
Facilities Council [grant number 2615874].

\appendix

\section{Casimir eigenvalue corresponding to Poincaré}
\label{AppPoincare}

The reason why the solutions to (\ref{eigen2}) provide representations
of the Lorentz algebra is the following.\cite{Delmastro:2017erq}
Assume that the $\{\omega_\ell\}$ provide a complete set of solutions
to
\begin{equation}
  -C_2\omega_\ell = \alpha \omega_\ell,
  \label{eigen3}
\end{equation}
where, in general dimension, $\ell$ is a multi-index. Acting with a
Lorentz generator $M_{\mu\nu}$ on (\ref{eigen3}) one has
\begin{equation}
  -M_{\mu\nu} C_2\omega_\ell = \alpha M_{\mu\nu}\omega_\ell,
\end{equation}
and, because $M_{\mu\nu}$ and $C_2$ commute,
\begin{equation}
  -C_2 M_{\mu\nu}\omega_\ell = \alpha M_{\mu\nu}\omega_\ell,
\end{equation}
which means that $M_{\mu\nu}\omega_\ell$ is also a solution of
(\ref{eigen3}). Now, since the $\{\omega_\ell\}$ are all the solutions
to the eigenvalue equation (\ref{eigen2}), one has necessarily that
\begin{equation}
  M_{\mu\nu}\omega_\ell = \sum_{\ell'} a_{\ell;\mu\nu}^{\ell'} \omega_{\ell'}, 
\end{equation}
which indicate that indeed the $\omega_\ell$ provide a representation
of the $M_{\mu\nu}$ via the matrices $K_{\mu\nu}$ with elements
\begin{equation*}
  (K_{\mu\nu})_\ell^{\ell'}=a_{\ell;\mu\nu}^{\ell'},  
\end{equation*}
for each $\mu, \nu$.

That the eigenvalue $-\alpha=d-1$ in (\ref{Casim1}) corresponds to an algebra which contains   Poincaré can be proved in general in any dimension by considering the algebra commutator
\begin{equation}
  [M_{\mu\nu},P_\rho] = i (\eta_{\rho\nu}P_\mu-\eta_{\rho\mu}P_\nu),
\end{equation}
which yields\footnote{For the interpretation that follows, it is
  important to obtain an expression with the $M$ to the right of the
  $P$. One can also get $[M^{\mu\nu}M_{\mu\nu},P_\rho] = 4i
  M_{\mu\rho}P^\mu - 2(d-1)P_\rho $.}
\begin{equation}
  [M^{\mu\nu}M_{\mu\nu},P_\rho] = 4iP^\mu M_{\mu\rho} + 2(d-1)P_\rho,
\end{equation}
and hence, if $C_2=1/2M^{\mu\nu}M_{\mu\nu}$,
\begin{equation}
  [C_2,P_\mu] = 2i P^\nu M_{\nu\mu} + (d-1)P_\mu.
  \label{cCP}
\end{equation}
If we represent  the generators in terms of differential operators
(not necessarily acting on the mass-shell hyperboloid), $C_2$ is a
second-order operator and the first term on the right-hand side of
(\ref{cCP}) will be a pure first-order one, without a zeroth-order
contribution. From this it can be read that
\begin{equation}
  \hat{C}_2 P_\mu = (d-1) P_\mu,
\end{equation}
as stated.

\section{Algebra of conserved charges}
\label{AppCharges}

Consider two conserved charges, computed at $t=0$, in terms of the
corresponding differential operators acting on the Fourier
coefficients of the scalar field,
\begin{equation}
  \begin{split}
    P &= \int\dk \bar{a}(\vk)\hat{P}a(\vk),\\
    Q &= \int\dk \bar{a}(\vk)\hat{Q}a(\vk).
  \end{split}
\end{equation}
Using the Poisson brackets of the Fourier modes, it can be shown,
without resorting to any integration by parts, and hence without
having to consider boundary contributions, that
\begin{equation}
  \{P,Q\} = -i \int\dk \bar{a}(\vk)[\hat{P},\hat{Q}]a(\vk),
\end{equation} 
which shows that if the algebra of the differential operators does not
exhibit any central extension, neither does the Poisson algebra of the
charges.

\bibliographystyle{utphys}
\bibliography{BMSalgs}

\end{document}